\documentclass[10pt,twocolumn,twoside,english]{IEEEtran}

\usepackage{amsthm}
\usepackage{amsfonts} 
\usepackage{amssymb}
\usepackage{latexsym} 
\usepackage{cite}
\usepackage[cmex10]{amsmath}
\usepackage{algorithmic}
\usepackage{upgreek}
\usepackage{url}
\usepackage[tight,footnotesize]{subfigure}
\usepackage{multirow}

\usepackage{siunitx}

\usepackage[pdftex]{graphicx}
\usepackage{color}
\usepackage{soul}
\definecolor{mgreen}{rgb}{0,0.2,0.2}
\definecolor{mblue}{rgb}{0,0,0.8}
\definecolor{mred}{rgb}{0.8,0,0}

\newcommand{\txblue}[1]{{\color{black}{#1}}}

\usepackage{array}
\newcolumntype{L}[1]{>{\raggedright\let\newline\\\arraybackslash\hspace{0pt}}m{#1}}
\newcolumntype{C}[1]{>{\centering\let\newline\\\arraybackslash\hspace{0pt}}m{#1}}
\newcolumntype{R}[1]{>{\raggedleft\let\newline\\\arraybackslash\hspace{0pt}}m{#1}}

\newcommand{\rrn}{r_r}
\newcommand{\erfc}{\mathrm{erfc}\,}
\newcommand{\rxi}[1]{\mathrm{Rx}#1}

\usepackage{amsthm}

\newtheorem{theorem}{Theorem}

\newtheorem{corollary}[theorem]{Corollary}

\newtheorem{lemma}[theorem]{Lemma}

\begin{document}

\title{Molecular MIMO: From Theory to Prototype}

\author{\IEEEauthorblockN{Bon-Hong~Koo,~\IEEEmembership{Student~Member,~IEEE,} Changmin~Lee,~\IEEEmembership{Student~Member,~IEEE,} H.~Birkan~Yilmaz,~\IEEEmembership{Member,~IEEE,} Nariman Farsad,~\IEEEmembership{Member,~IEEE,} Andrew Eckford,~\IEEEmembership{Senior~Member,~IEEE,}~and\\ Chan-Byoung Chae,~\IEEEmembership{Senior~Member,~IEEE}}
\thanks{B. Koo, C. Lee, H. B. Yilmaz and C.-B. Chae (Corresponding Author) are with the Yonsei Institute of Convergence Technology, School of Integrated Technology, Yonsei University, Korea. Email: \{harpeng7675, cm.lee, birkan.yilmaz, cbchae\}@yonsei.ac.kr. N.~Farsad is with Stanford University, USA. Email: nfarsad@stanford.edu. A.~Eckford are with the Department of Electrical and Computer Science, York University, Canada. Email: aeckford@cse.yorku.ca. }
\thanks{This research was supported in part  by the MSIP (Ministry of Science, ICT and Future Planning), Korea, under the ``IT Consilience Creative Program" (IITP- 2015-R0346-15-1008) supervised by the IITP (Institute for Information \& Communications Technology Promotion) and by the Basic Science Research Program (2014R1A1A1002186) funded by the Ministry of Science, ICT and Future Planning (MSIP), Korea, through the National Research Foundation of Korea.}
\thanks{Full demo video is available at http://www.cbchae.org/.}
}


\maketitle

\begin{abstract}
In diffusion-based molecular communication, information transport is governed by diffusion through a fluid medium. The achievable data rates for these channels are very low compared to the radio-based communication system, since diffusion can be a slow process. To improve the data rate, a novel multiple-input multiple-output (MIMO) design for molecular communication is proposed that utilizes multiple molecular emitters at the transmitter and multiple molecular detectors at the receiver (in RF communication these all correspond to antennas). Using particle-based simulators, the channel's impulse response is obtained and mathematically modeled. These models are then used to determine inter-link interference (ILI) and inter-symbol interference (ISI). It is assumed that when the receiver has incomplete information regarding the system and the channel state, low complexity symbol detection methods are preferred since the receiver is small and simple. Thus four detection algorithms are proposed---adaptive thresholding, practical zero forcing with channel models excluding/including the ILI and ISI, and Genie-aided zero forcing. The proposed algorithms are evaluated extensively using numerical and analytical evaluations. 
\end{abstract}


\IEEEpeerreviewmaketitle

\section{Introduction} 
\IEEEPARstart{T}{iny} devices that exhibit strong cooperation via communication on a small scale are capable of having an impact at the macro-scale. In fact, communication on a small scale (micro and possibly nano) is at a critical juncture, where there is a need to engineer communication systems at these scales~\cite{nakano2013molecularC}. A new paradigm called {\em nanonetworking}, focuses on this problem.
At the micro and nano scale, electromagnetic communication is challenging because of constraints such as the ratio of the antenna size to the wavelength of the electromagnetic signal~\cite{akyildiz2008nanonetworksAN,nakano2013molecularC,farsad2014comprehensiveSO_ARXIV}. Thus taking their cue from nature, researchers have focused their attention on molecular communication where information is carried through chemical signals
~\cite{atakan2012bodyAN, nakano2012molecularCA,kim2013novelMT,kuran2010energyMF}. In the literature, numerous molecular communication systems have been proposed: molecular communication via diffusion (MCvD) (such as calcium and pheromone signaling), motor proteins motility over microtubules, microtubule motility over stationary motor proteins, and  bacterium-based communication~\cite{nakano2013molecularC,kuran2010energyMF,far15NANO,lio2012opportunisticRT}.  

An MCvD system consists of three main components: a transmitter, a fluid environment between the transmitter and the receiver, and a receiver. The transmitter and the receiver must have at least one antenna\footnote{The term `antenna' we use in this paper can physically be a bulge. In Fig.~\ref{Fig:Topology_model}, bulges are presented as the extensions from the transmitter and the receiver node for interacting with the physical environment,and represent the molecular emitters and detector.} to send and receive chemical signals. To convey information, a transmitter node emits molecules into the medium via its transmit antennas. The emitted molecules then propagate through the environment according to the physical characteristics of the medium and molecules. The propagation is generally considered to be restricted to diffusion, unless, as in~\cite{srinivas2012molecularCI,kim2014symbolIO}, the environment has flow. Some of the propagating molecules arrive at the receive antenna or antennas where they trigger a reception process.
A macro-scale counterpart to such a system was introduced by Farsad \emph{et al.} in \cite{farsad2013tabletopMC} as a tabletop testbed, and  the nonlinearity in the channel was modeled as Gaussian noise~\cite{farsad2014channelAN}.

Some of the system models in the literature, assume that molecules are removed from the environment after hitting the receive antenna~\cite{srinivas2012molecularCI,yilmaz2014_3dChannelCF}. This can be modeled as the first-hitting process. The first-hitting process dictates that each molecule can contribute to the signal only once. On the other hand, some models omit the first-hitting process,  allowing molecules to pass through the receive antenna with no alteration to the movement dynamics (i.e., the molecules can go inside and out of the receiver antenna multiple time)
	~\cite{pierobon2013capacityOA,mahfuz2014strengthBO,noel2014improvingRP}. Furthermore, analytical models were developed for arrival, receptor dynamics, degradation, and noise~\cite{yilmaz2014arrivalMF,akkaya2015effectOR,heren2014effectOD,noel2014improvingRP,pierobon2011diffusionBN}.

From a communications perspective, the fundamental processes of an MCvD system are modulation, emission, signal propagation, reception, and demodulation. A circuit-based model of the system was introduced in~\cite{chou2014molecularCN}. In the literature, information is modulated on various aspects of the messenger molecules, such as molecule concentration, molecule identity, concentration ratio, and signal frequency~\cite{srinivas2012molecularCI,kuran2011modulationTF,kim2013novelMT,yilmaz2014simulationSO_SIMPAT}. The main challenge, researchers have pointed out, 
 is 
the heavy tail nature of the molecular signal that causes severe inter-symbol interference (ISI) and the slow propagation~\cite{tepekule2014isiMT,noel2014improvingRP}. To increase MCvD performance, the literature suggests many enhancements. These include incorporating ISI mitigation techniques, multiple molecule types, enzymes to reduce ISI, protrusions to increase reception probability, and pre-equalization methods~\cite{kim2014symbolIO,noel2014improvingRP,yilmaz2014simulationSO_SIMPAT,tepekule2014isiMT,genc2013receptionEW,tepekule2015novelPE,meng2014receiverDF}. 

The present capabilities of molecular communications are rather primitive. Indeed, the world's first artificially engineered molecular communication (even with drift) introduced in~\cite{farsad2013tabletopMC,farsad2014channelAN} has a \emph{chemical efficiency} of 0.3 bits/s/chemical over a free-space distance of a few meters. For commercial applications, such a rate might be too slow. A good candidate for resolving this problem is the utilizing of the multiple-input multiple-output (MIMO) technique~\cite{meng2012mimoCB}. The authors in~\cite{meng2012mimoCB}, however, mainly focused \txblue{on} multiuser interference and assumed ISI to be negligible (i.e., the authors \txblue{paid little attention} to ISI in the system model). 

Therefore, we introduce a MIMO system model that includes ISI and inter-link interference (ILI) and the detection algorithms specific to molecular MIMO systems. To the best of our knowledge, this is \emph{the first work that analytically considers ISI and ILI for a molecular MIMO system}. Its novelty goes beyond the SISO work in that the presence of
additional antennas that interfere with each other brings about the need for a new channel model.
To demonstrate the practicality of our models, we also present measurement results via a novel tabletop testbed.  

We first model the channel's impulse response by using a channel function similar to the single-input single-output (SISO) system. By utilizing the developed MIMO simulator, we model the channel and obtain the parameters of the modified channel impulse function.\footnote{The fully-featured 3-D molecular MIMO simulator is available at http://www.cbchae.org/.} Consequently, in a molecular MIMO system, we utilize the channel model to model the ILI that originates from the other antenna at the current symbol slot and the ISI from previous emissions. \txblue{The Brownian motion of individual molecules was simulated to generate the channel statistics, and then these statistics were used when simulating the detection algorithms.} We analyze the system performance in terms of bit error rate (BER) and channel model via signal distributions. For this work, we chose to focus on an $2\times2$ molecular MIMO system. However, similar analysis can be 
	applied to
	an $M\times N$ molecular MIMO system with appropriate channel modeling. 

The main contributions of this paper are as follows:
\begin{itemize}
	\item \textbf{Modeling MIMO Systems:} We model molecular MIMO systems that consider ISI and ILI \txblue{by} utilizing the channel response function.
	\item \textbf{Formulating ISI and ILI:} We derive formulations for ISI and ILI.
	\item \textbf{Molecular MIMO Signal Detection Algorithms:} Utilizing a single type of 
	molecule,
	we propose detection algorithms specific to molecular MIMO systems and analyze the performance in terms of bit error rate (BER).
	\item \textbf{Tabletop Molecular MIMO Testbed:} We implement our proposed algorithms on our 
	macro-scale tabletop molecular MIMO testbed
	and perform measurements to evaluate the system's performance.
\end{itemize}

The rest of this paper is organized as follows. In Section~\ref{model}, we describe the system model, including topology, propagation, and communication model. In Section~\ref{fit_n_algorithm}, we detail the channel estimation method with channel modeling. Section~\ref{Sec:Theory} presents the theoretical analysis \txblue{of} the proposed system. Hardware realizations of the theoretical works are covered in Section~\ref{Sec:Testbed}. Section~\ref{Sec:results} shows the numerical results and Section~\ref{Sec:conclusion} presents our conclusions.

\section {System Model}
\label{model}
The molecular communication system considered in this paper is a 3-D environment with two point sources and two spherical receive antennas. The transmitter releases a certain number of messenger molecules at once without a biased direction. The emitted molecules travel via diffusion inside the fluid medium. We ignore the effect of collisions among diffusing molecules and assume that the enzymatic reaction is treated as an additive noise in the theoretical analysis. When a molecule hits the boundary of one of the spherical antennas, it is immediately absorbed by the receiving antenna and removed from the medium. We assume that the transmitter-receiver pair has synchronized symbol slots and the receiver can, during a symbol slot, count the number of received molecules.

During communications, the transmit antennas convey independent messages to their corresponding receive antennas. Both links of the molecular MIMO system use the same type of molecule. Hence the molecules from the other transmit antenna may cause ILI.

\subsection{Topology and Propagation Model}
\label{Sec:topology}

\begin{figure}
	\centering
	\includegraphics[width=0.99\columnwidth,keepaspectratio]
	{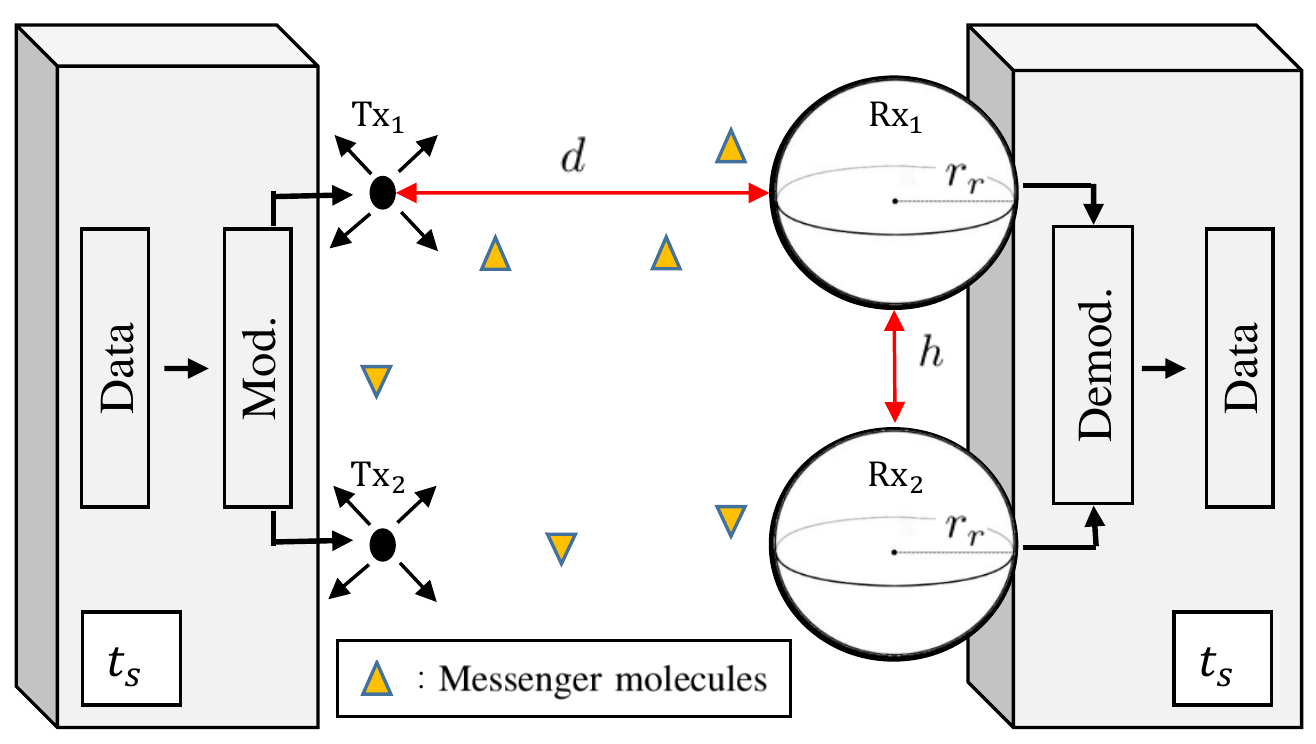}
	\caption{Topological model of a molecular $2\times2$ MIMO system. The modulation process individually modulates the information on to a physical property of the messenger molecules. On the other hand, the demodulation process demodulates the received signal separately.} 
	\label{Fig:Topology_model}
\end{figure}

As shown in Fig.~\ref{Fig:Topology_model}, there are two point antennas Tx$_1$ and Tx$_2$ at transmitter Tx, placed $d$ distance apart from the corresponding spherical antennas Rx$_1$ and Rx$_2$ at the receiver Rx. The two receive antennas have the same radius $\rrn$, and are placed $h$ distance apart. \txblue{Their centers and those of the transmit antennas form a rectangle on a plane.} The spherical antennas are attached to the receiver body and we assume that just the antennas are capable of receiving molecules, a tenable assumption given the many examples found in nature. For example, epithelial cells, neurons, and migrating cells are polarized cells having heterogeneous receptor deployments \cite{luddens1995receptor}, an adaptation to the environment and the signaling mechanism. After the reception, the receiver demodulates the signal using a corresponding algorithm. The candidate algorithms are presented in Section~\ref{Sec:Detection_algorithms}. The transceiver pair is assumed to have synchronized time duration between consecutive symbols denoted by $t_s$. Many examples of synchronization among bacteria are observed in nature and inspired from it. For example, a simple synchronization algorithm for communicating pairs was proposed in~\cite{shah2013blindSI}.

When a molecule is released from a point source, its movement in a fluid is governed by diffusion and drift (drift is applicable if there is a flow). The dynamics of diffusion can be described by Brownian motion. The \txblue{expected} fraction of absorbed molecules with respect to time, $F(t|r_r,d,D)$, is derived in \cite{yilmaz2014_3dChannelCF} when the receiver is a spherical body and the transmitter is a point source in a 3-D environment. The derived formula determines \txblue{the expected number of hitting molecules and hence} the channel characteristics. The formula is as follows:
\begin{equation}
\label{eqn_first_hitting_3d}
F(t|\rrn,\, d,\, D)= \frac{\rrn}{\rrn+d}\, \erfc\left( \frac{d}{\sqrt{4Dt}} \right)
\end{equation}
where $D$ denotes the diffusion coefficient. The two receive antennas in our setup, however, prevent the direct use of (\ref{eqn_first_hitting_3d}) since hitting Rx$_1$ and Rx$_2$ become dependent events. The closed form expression for the MIMO setup is still an open problem. Therefore, for channel modeling, we simulate the movement of released particles within the given MIMO setup by using
\begin{align}
\begin{split}
(x_{t}, y_{t}, z_{t}) &= (x_{t-\Delta t}, y_{t-\Delta t}, z_{t-\Delta t}) + (\Delta x, \Delta y, \Delta z),\\
\Delta x & \sim \mathcal{N}(0, 2D\Delta t), \\
\Delta y & \sim \mathcal{N}(0, 2D\Delta t), \\
\Delta z & \sim \mathcal{N}(0, 2D\Delta t) 
\end{split}
\label{eqn_propagation_model}
\end{align}
where $x_{t}$, $y_{t}$, $z_{t}$, \txblue{$\Delta t$,} and $\mathcal{N}(\mu, \sigma^2)$ are the particles' positions at each dimension at time $t$, \txblue{the simulation time step,} and the normal distribution with mean $\mu$ and variance $\sigma^2$. The Brownian motion simulator for the MIMO setup is a modified version of the  simulator that was developed for a SISO case in~\cite{yilmaz2014simulationSO_SIMPAT}. During its trip, if a molecule hits one of the spherical receive antennas, then it is absorbed and removed from the environment. A molecule can contribute to the signal just once and we have four different $F(t|r_r, d, D)$ values depending on the molecule's emission source and hitting antenna.  We use the simulation data to formulate $F(t|r_r, d, D)$ for the 2$\times$2 molecular MIMO setup and utilize it for the analysis. \txblue{Note that we have two different $F(t|r_r, d, D)$ values due to symmetry, which enables us to combine equivalent ones to have better estimates from simulation data. }

 \subsection{Communication Model}
\label{communication}	

\txblue{Each receive antenna has the closer transmit antenna as its pair source of communication, while the other transmit antenna becomes a non-pair source.} To encode information, we use the binary concentration shift keying (BCSK) modulation technique~\cite{kim2013novelMT}. We denote $Q_1$ and $Q_0$ as the number of molecules released to send bit-1 and ${\mbox{bit-0}}$. In this paper, the transmitter sets ${Q}_{0}$ as zero to reduce the energy consumption and separate, as much as possible, the signal amplitudes. The transmitter has independent sets of bits $x_1$ and $x_2$ for their own messages. During the $m^{\mathrm{th}}$ symbol, Tx$_1$ and Tx$_2$ send $x_1[m]$ and $x_2[m]$ each by releasing ${Q}_{1}\cdot{x}_{1}[m]$ and ${Q}_{1}\cdot{x}_{2}[m]$ molecules at the start of the symbol time, and wait until the next emission time. \txblue{This representation is valid for both $x_i[m]=0$ and $x_i[m]=1$, since $Q_{0}=0$.} The number of molecules absorbed at the receiver follows a binomial distribution with a hitting probability, which is related to $t_s$, $d$, $r_r$, and $D$~\cite{kuran2010energyMF}. We define $F_{ij}(t_1, t_2)$ as the probability of a molecule hitting Rx$_i$ when the molecule is  released from a Tx$_j$ \txblue{between time $t_1$ and $t_2$ after the transmission}. So we can define the random variable $\mathcal{S}_{ij}(t_1,t_2)$ as follows:
\begin{equation}
\label{eq_binom_rv_sij}
\mathcal{S}_{ij}(t_1,t_2) \triangleq \mathcal{B}\left(Q_1,F_{ij}(t_1,t_2)\right),
\end{equation}
where $\mathcal{B}(N,p)$ denotes a binomial random variable with $N$ trials and success probability $p$. $\mathcal{S}_{ij}$ is utilized while evaluating the number of received molecules at Rx$_i$ that originates from Tx$_j$.

The number of molecules received by Rx$_i$ during the $m^{\mathrm{th}}$ time slot is denoted by $y_{\mathrm{Rx}_i}[m]$. The received molecules are composed of molecules emitted from the two transmit antennas, at the current and the previous symbol slots. $y_{\mathrm{Rx}_i}[m]$ can be expressed as the sum of the number of molecules released from \textit{1) the pair source at the current symbol slot}, \textit{2) the pair source at previous symbol slots}, \textit{3) the non-pair source at the current symbol slot}, and \textit{4) the non-pair source at the previous symbol slots}. Considering the events of false-capture or mis-capture\footnote{The false-capture indicates a positive noise that can be occurred by the invasion of outer molecules, which are indistinguishable from the target molecules by the receiver. The mis-capture indicates a negative noise in case of molecular decomposition.} by the receiver that affects the number of receptions at Rx$_i$, we add the noise term $n_{i}[m]$ which is commonly assumed to be a normal distribution $\mathcal{N}(0,\sigma_{n}^{2})$. It should be noted that the added noise term is not the diffusion noise and in our model the random variables $\mathcal{S}_{ij}$ include the randomness of the diffusion.

To simplify the modeling, we denote $\mathcal{S}_{ij}[k]$ as follows:
\begin{equation}
\label{eq_binom_rv_s_k}
\mathcal{S}_{ij}[k] \triangleq \mathcal{S}_{ij}(kt_{s},(k+1)t_{s})
\end{equation}
to indicate the probability of a molecule emitted from a transmitter being captured by the receiver during the $k^{th}$ time slots after the emission. In the representation, $k\!=\!0$ indicates the case where molecules came from the current symbol, while adding all the cases of $k\!>\!0$ indicates the case of molecules from past symbols. $i\!=\!j$ indicates the case of molecules from the corresponding antenna, while $i\neq{j}$ indicates the molecules from the other link. In our configuration of a $2\!\times\!2$ MIMO system, $i$ and $j$ can be either $1$ or $2$. As a result, $y_{\mathrm{Rx}_i}[m]$ yields \eqref{eq_y_Rxi_n} provided that $i\neq{j}$ holds. \txblue{We present \eqref{eq_y_Rxi_n} at the top of the next page.}
\begin{figure*}
	\begin{align}
	\label{eq_y_Rxi_n}
	y_{\mathrm{Rx}_i}[m]=\underbrace{\mathcal{S}_{ii}[0]\cdot{x}_{i}[m]}_\text{desired signal}+\underbrace{\mathcal{S}_{ij}[0]\cdot{x}_{j}[m]}_\text{ILI}+\underbrace{\sum\limits_{k=1}^{m-1}\left(\underbrace{\mathcal{S}_{ii}[k]\cdot{x}_{i}[m-k]}_\text{ISI from Tx$_i$}+\underbrace{\mathcal{S}_{ij}[k]\cdot{x}_{j}[m-k]}_\text{ISI from Tx$_j$}\right)}_\text{total ISI}+\underbrace{n_{i}[m]}_\text{noise}.
	\end{align}
	\hrulefill
\end{figure*}

We write \eqref{eq_y_Rxi_n} in matrix form for each symbol slot in two different forms by excluding and including the ILI term in the channel matrix. \txblue{The channel model excluding ILI for a} $2\times2$ MIMO system is given as:
\begin{align}
\label{eq_y_excludingI_inMatrixForm}
\begin{split}
\begin{bmatrix} y_{\mathrm{Rx}_1} \\ y_{\mathrm{Rx}_2}\end{bmatrix} \!&=\! \begin{bmatrix} \mathcal{S}_{11}[0] & 0\\ 0& \mathcal{S}_{22}[0]\end{bmatrix} \begin{bmatrix} x_1 \\ x_2\end{bmatrix} \!+\! \begin{bmatrix} {\zeta}_{1}+{I}_{1} \\ {\zeta}_{2}+{I}_{2}\end{bmatrix}
\end{split}
\end{align}
where $\zeta_i$ and $I_i$ denote the `ILI' term and sum of the `total ISI' term and the `noise' term respectively from the formulation~\eqref{eq_y_Rxi_n}.
\txblue{Note that the model \eqref{eq_y_excludingI_inMatrixForm} still accounts for the ILI term in the interference vector. The name 'excluding' comes from the fact that the ILI term is excluded from the channel matrix.}
We write \eqref{eq_y_excludingI_inMatrixForm} in short as
\begin{equation}
\pmb{y}=\pmb{H}_{\mathrm{ex}}\pmb{x}+\pmb{I}_{\mathrm{ex}}.
\end{equation}
The channel model including the ILI in $\pmb{H}_{\mathrm{in}}$ is given as:
\begin{align}
\label{eq_y_includingI_inMatrixForm}
\begin{split}
\begin{bmatrix} y_{\mathrm{Rx}_1} \\ y_{\mathrm{Rx}_2}\end{bmatrix} \!=\! \begin{bmatrix} \mathcal{S}_{11}[0] & \mathcal{S}_{12}[0]\\ \mathcal{S}_{21}[0]& \mathcal{S}_{22}[0]\end{bmatrix} \begin{bmatrix} x_1 \\ x_2\end{bmatrix} \!+\! \begin{bmatrix} {I}_{1} \\{I}_{2}\end{bmatrix}
\end{split}
\end{align}
and it is written in short as
\begin{equation}
\pmb{y}=\pmb{H}_{\mathrm{in}}\pmb{x}+\pmb{I}.
\end{equation}
The details of $\pmb{I}$ and $\pmb{I}_{\mathrm{ex}}$ are presented in Section~\ref{Sec:threshold}.

\section{Channel Modeling and Proposed Detection Algorithms}
\label{fit_n_algorithm}
 \txblue{In this section, we introduce our work to find} a formula for the first-hitting probability in the molecular MIMO setup to formulate ${F}_{ij}(t)$. Extensive simulations were carried out to understand the underlying \txblue{model} and we use nonlinear curve fitting on the simulation data. After obtaining the fitted cumulative distribution functions (CDFs) we \txblue{carry} out our analytical derivations by utilizing approximated functions ${F}_{ij}(t)$. \txblue{We then present} the proposed signal detection algorithms.

\subsection{Channel Modeling}
\label{fitting}

We use a model function that is similar to~\eqref{eqn_first_hitting_3d} for fitting the simulation data (i.e., the formula is similar with molecular SISO in a 3-D environment with some control coefficients). \txblue{A similar model function is used in \cite{farsad2014channelAN} for modeling noise effects in the SISO testbed and in \cite{koo2015detectionAF} for modeling molecular MIMO channel via simulations. Control coefficients are selected to comprehend the effect on amplitude that stems from diffusion coefficient and time.} The model function structure for nonlinear fitting is as follows:
\begin{equation}
	F_{ij}^{\text{model}}(t|\rrn,\, d_{ij},\, D) =  \frac{b_1\,\rrn}{d_{ij}\!+\!\rrn}  \erfc\left( \frac{d_{ij}}{{(4D)^{b_2} \,t^{b_3}}} \right)
	\label{eqn_fpp_model_function}
\end{equation}
where $b_1$, $b_2$, and $b_3$ are fitting variables and $d_{ij}$ is the distance between Tx$_j$ and Rx$_i$. We run extensive simulations for 32 different parameter sets and estimate the mean CDF of the hitting molecules. In each simulation, 5000 molecules per antenna are emitted and we carry out 500 simulations. In the simulations, the hitting molecules are separated according to where they originated for finding ${F}_{11}$, ${F}_{12}$, ${F}_{21}$, and ${F}_{22}$. Due to the symmetry of the topology, ${F}_{11}$ and ${F}_{22}$ and similarly ${F}_{12}$ and ${F}_{21}$ are very close to one another (so we assume ${F}_{11} = {F}_{22}$ and ${F}_{12} = {F}_{21}$ due to symmetry). Also, note that the distance between Tx$_1$ and Rx$_2$ is different from the distance between Tx$_1$ and Rx$_1$. Therefore, the model function utilizes the corresponding distances for $F_{ij}$ as follows
\begin{align}
	\begin{split}
		d_{11} = d_{22} &= d,\\
		d_{12} = d_{21} &= \sqrt{(d+\rrn)^2 + (2\rrn+h)^2} \,- \rrn.
	\end{split}
\end{align}

\begin{figure}
	\centering
	\includegraphics[width=0.99\columnwidth,keepaspectratio]
	{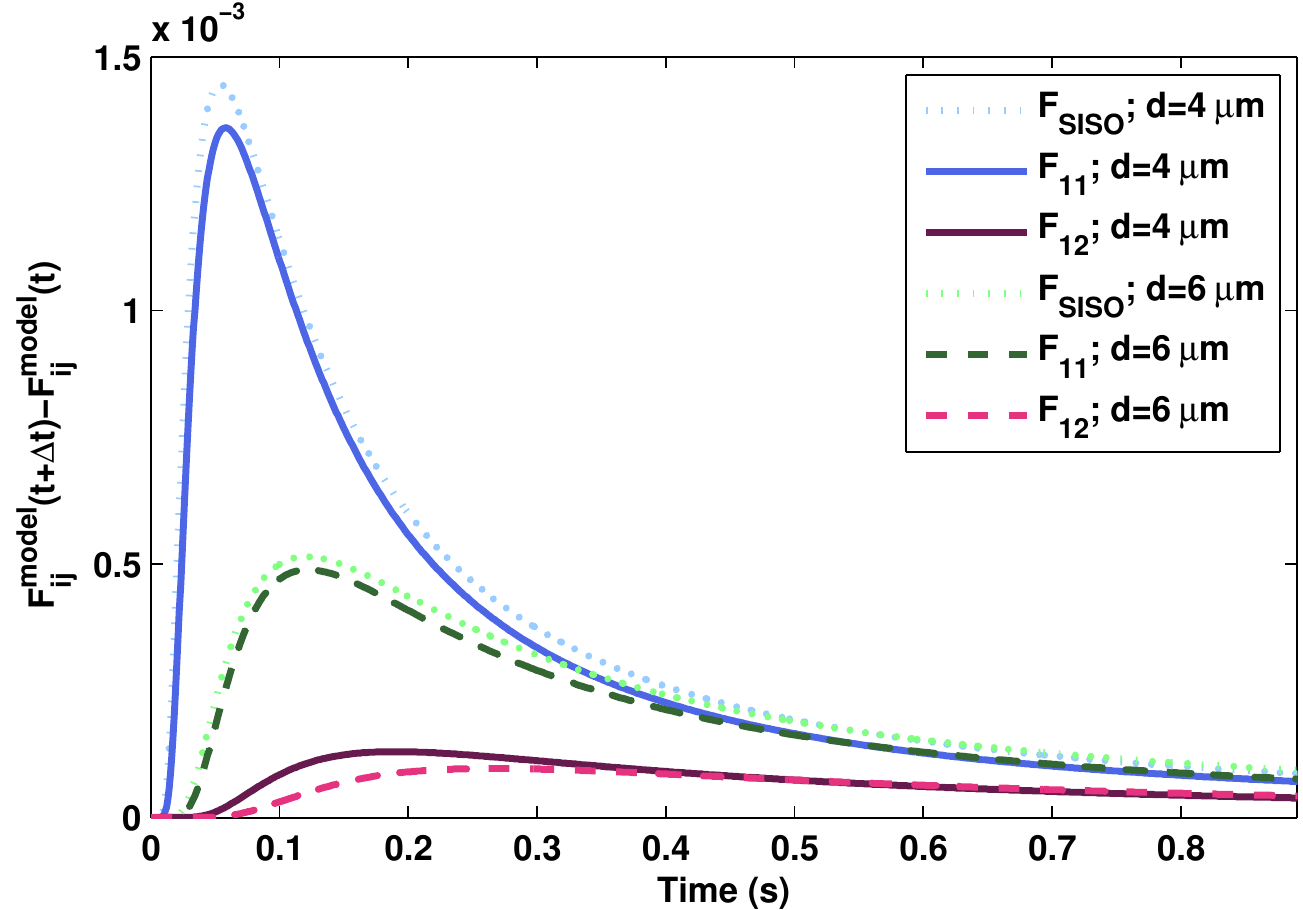}
	\caption{Fraction of hitting molecules for SISO cases and fitted fraction of hitting molecules for MIMO cases with an example parameter set ($D=\SI{50}{\micro\meter^2/\second}$, ${h=\SI{1}{\micro\meter}}$, ${\rrn=\SI{4}{\micro\meter}}$, ${\Delta t=\SI{0.001}{\second}}$).}
	\label{Fig:ft}
\end{figure}

\begin{figure*}[t]
	\begin{center}
		\subfigure[Fitted values of $b_1$ for $F_{11}$]
		{\includegraphics[width=0.32\textwidth,keepaspectratio]
			{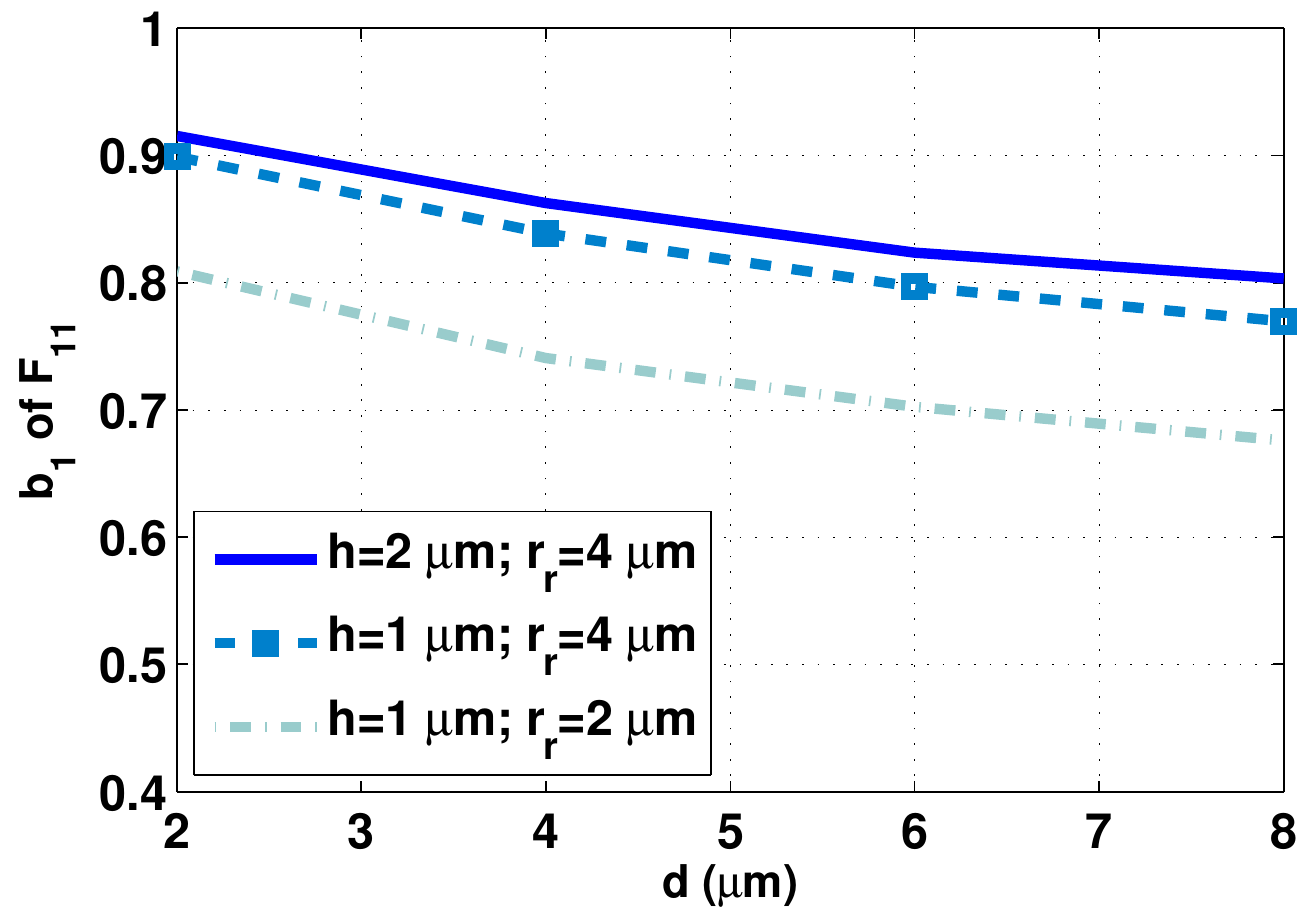}
			\label{subfig_fitting_b1} } 
		\subfigure[Fitted values of $b_2$ for $F_{11}$]
		{\includegraphics[width=0.32\textwidth,keepaspectratio]
			{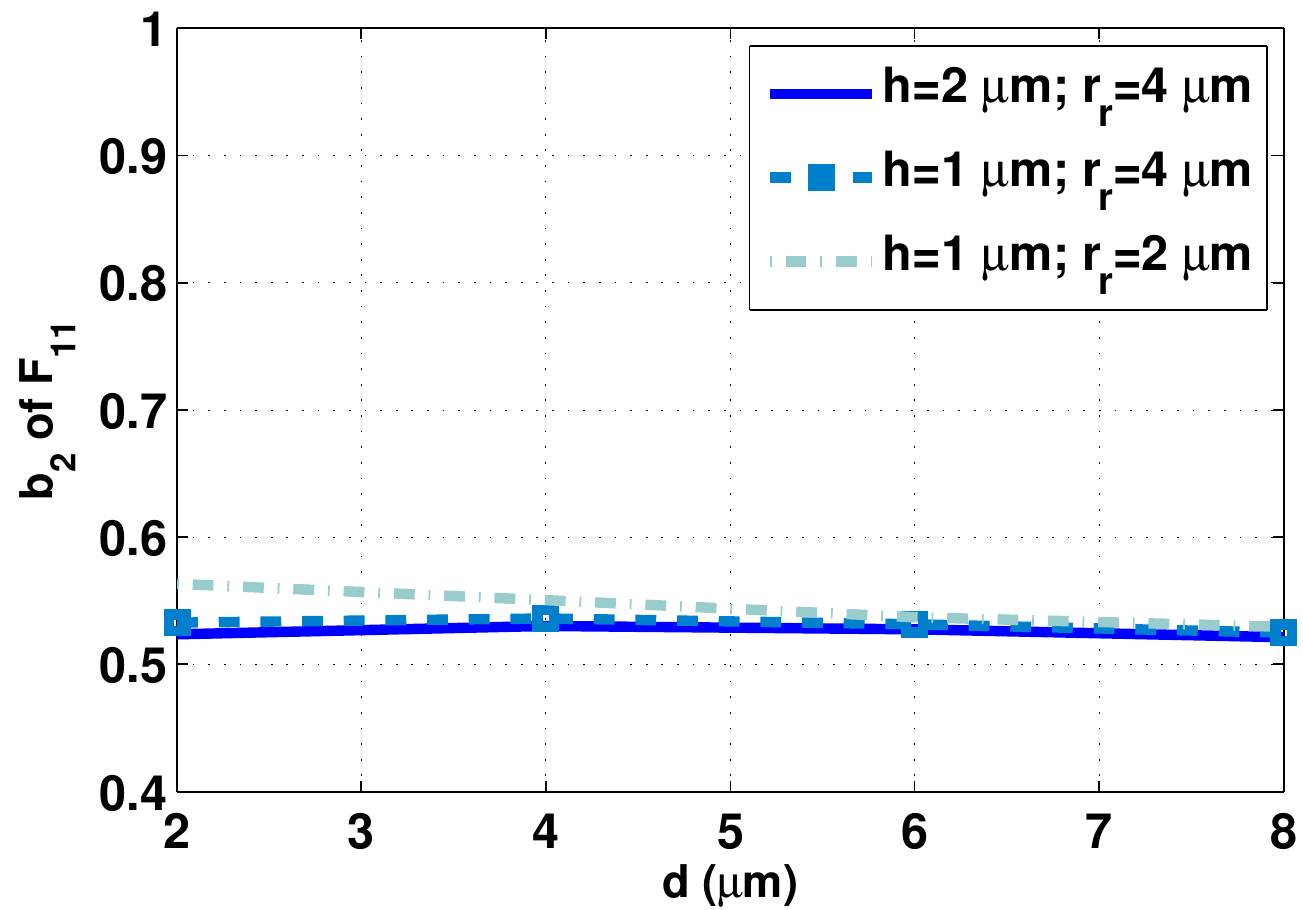}
			\label{subfig_fitting_b2} }
		\subfigure[Fitted values of $b_3$ for $F_{11}$]
		{\includegraphics[width=0.32\textwidth,keepaspectratio]
			{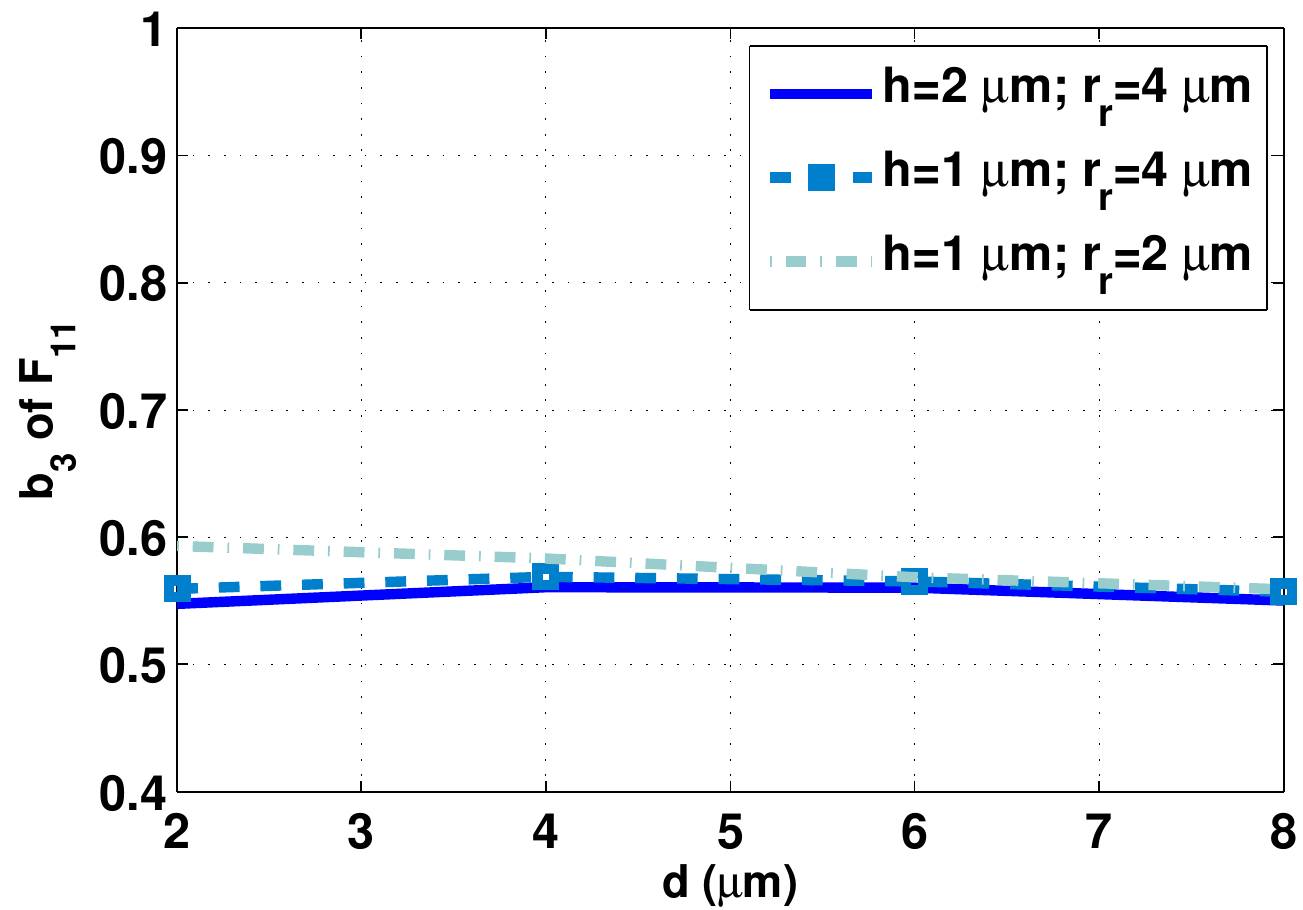}
			\label{subfig_fitting_b3} }
		\subfigure[Fitted values of $b_1$ for $F_{12}$]
		{\includegraphics[width=0.32\textwidth,keepaspectratio]
			{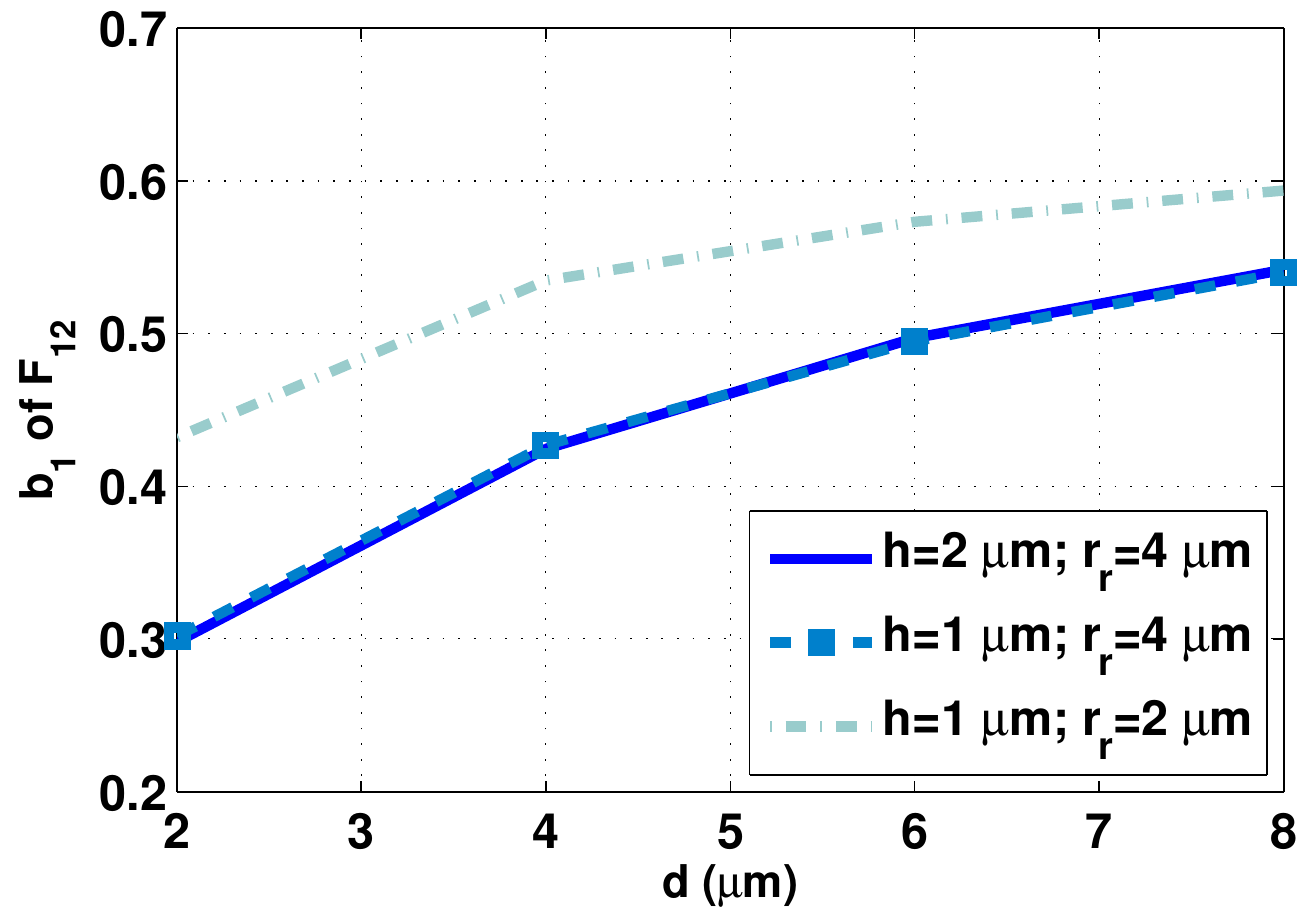}
			\label{subfig_fitting_cross_b1} } 
		\subfigure[Fitted values of $b_2$ for $F_{12}$]
		{\includegraphics[width=0.32\textwidth,keepaspectratio]
			{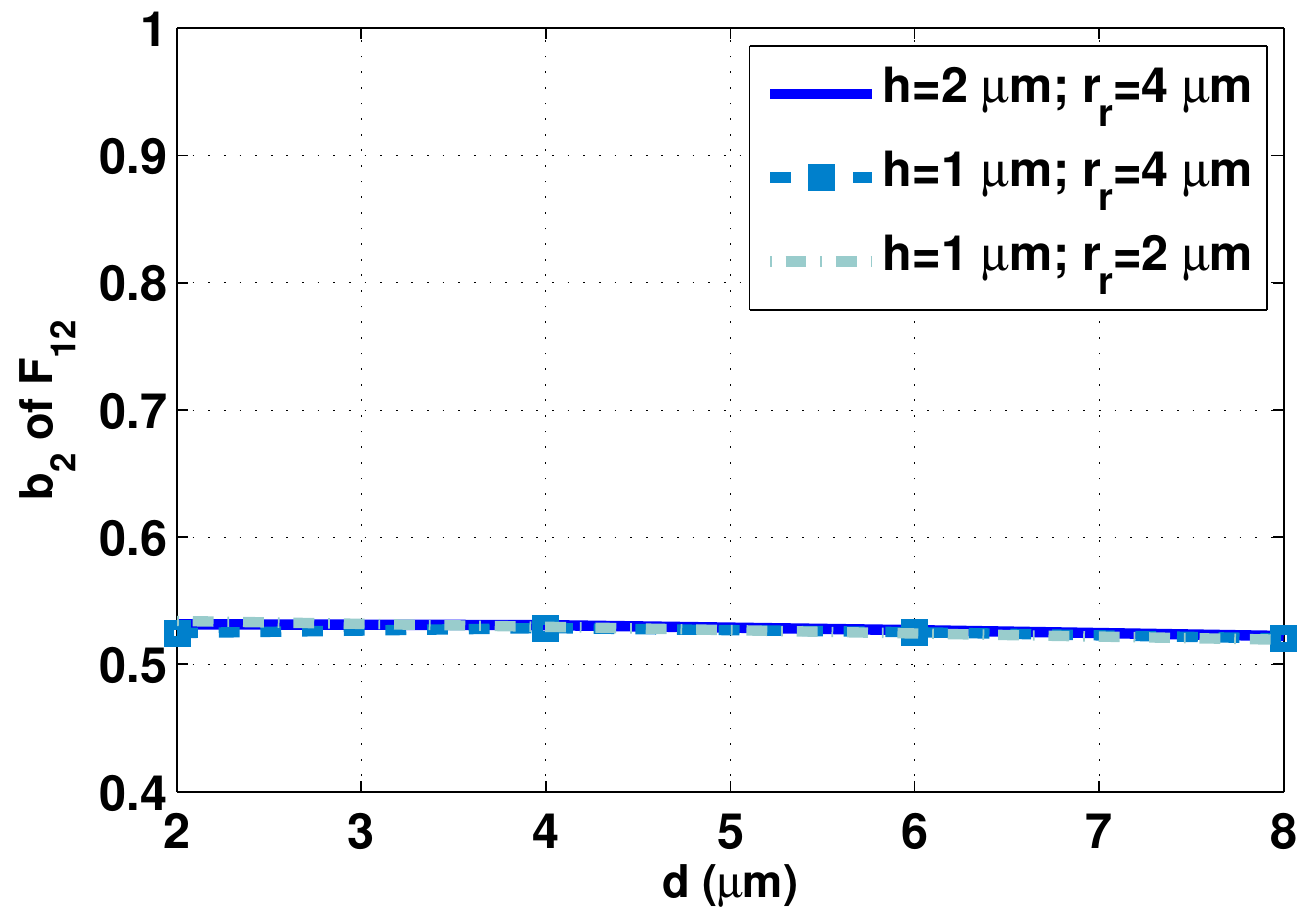}
			\label{subfig_fitting_cross_b2} }
		\subfigure[Fitted values of $b_3$ for $F_{12}$]
		{\includegraphics[width=0.32\textwidth,keepaspectratio]
			{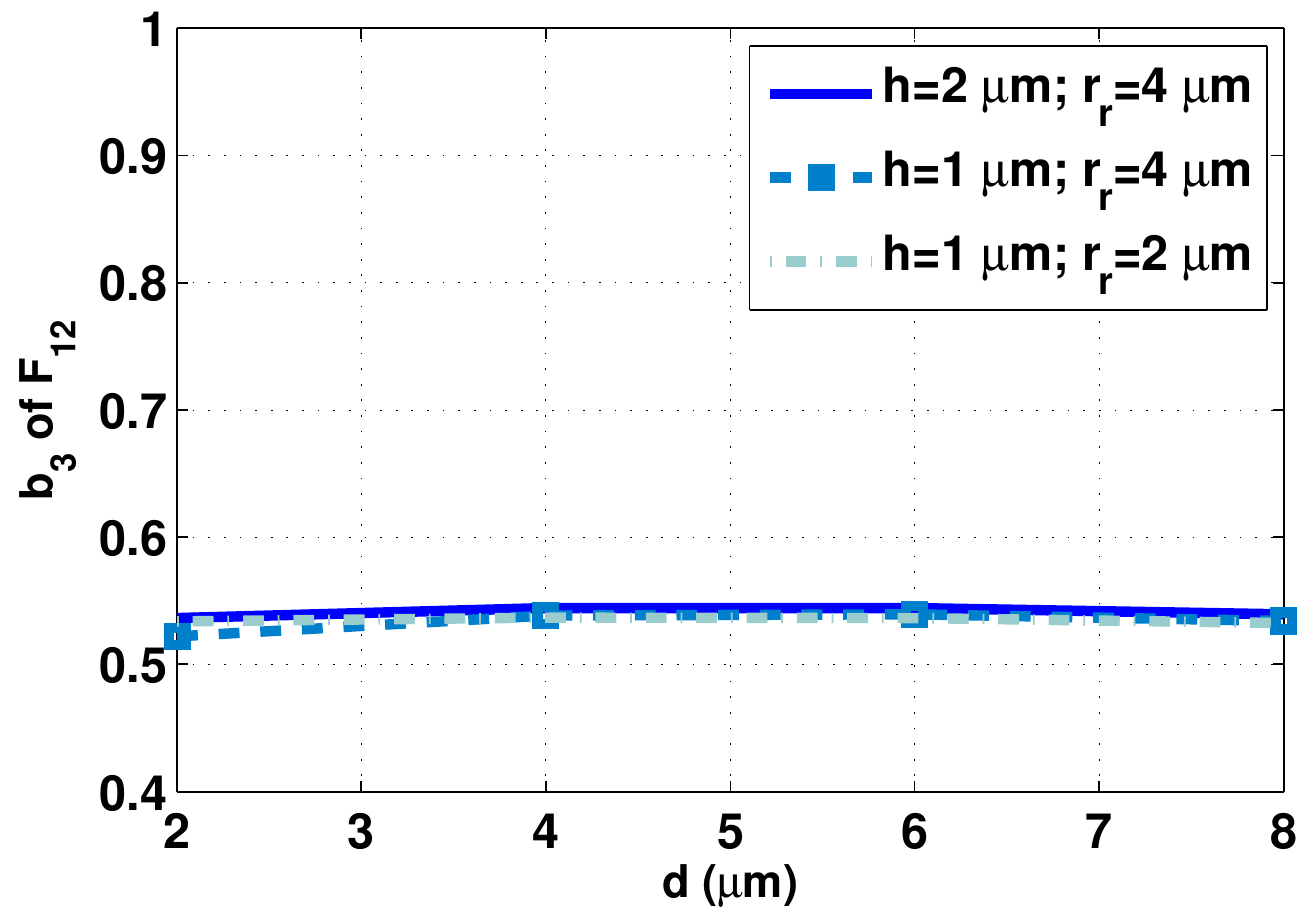}
			\label{subfig_fitting_cross_b3} }
	\end{center}
	\caption{Fitted model parameters of ${F}_{11}(t)$ and ${F}_{12}(t)$ for different $h$ and $\rrn$ values (\txblue{$D=\SI{50}{\micro\meter^2/\second}$}).}
	\label{fig_fitting_all}
\end{figure*}

The estimated CDFs are obtained from the simulation outputs by considering the histogram of the arriving molecules up to a sufficient time. We adopt 1.5 seconds which is enough to observe the peak of the signal under the given parameters. Starting from zero to the end of simulation time, we obtain a fraction of arriving molecules. Note that this function is monotonically increasing with time as expected from \eqref{eqn_fpp_model_function}. To fit nonlinear regression models to simulation data, this study uses the NonLinearModel class in MATLAB. We implement the model function and utilize the simulation outputs to have a closed form CDF estimation that obeys \eqref{eqn_fpp_model_function}.

Fig.~\ref{Fig:ft} shows the resulting signal function that defines the received signal at the intended receiver as $F_{11}$ and the ILI signal as $F_{12}$ with the given parameters. Increasing the distance results in shifting the peak amplitude time and reducing the peak amplitude value. The ILI term is non-negligible, but the distance has a more severe effect on the received signal.

Fig.~\ref{fig_fitting_all} presents distance-versus-fitted model variables for different $h$ and $\rrn$ values. The first and the second rows correspond to $F_{11}$ and $F_{12}$ modeling, respectively. The values of $b_2$ and $b_3$ change little when the distance is increased and, similar to the SISO case, they are close to 0.55 (in~\cite{yilmaz2014_3dChannelCF}, for the SISO case $b_2\!=\!b_3\!=\!0.5$). We also observe that $\rrn$ is more dominant than $h$.

\begin{figure}
	\centering
	\includegraphics[width=0.99\columnwidth,keepaspectratio]
	{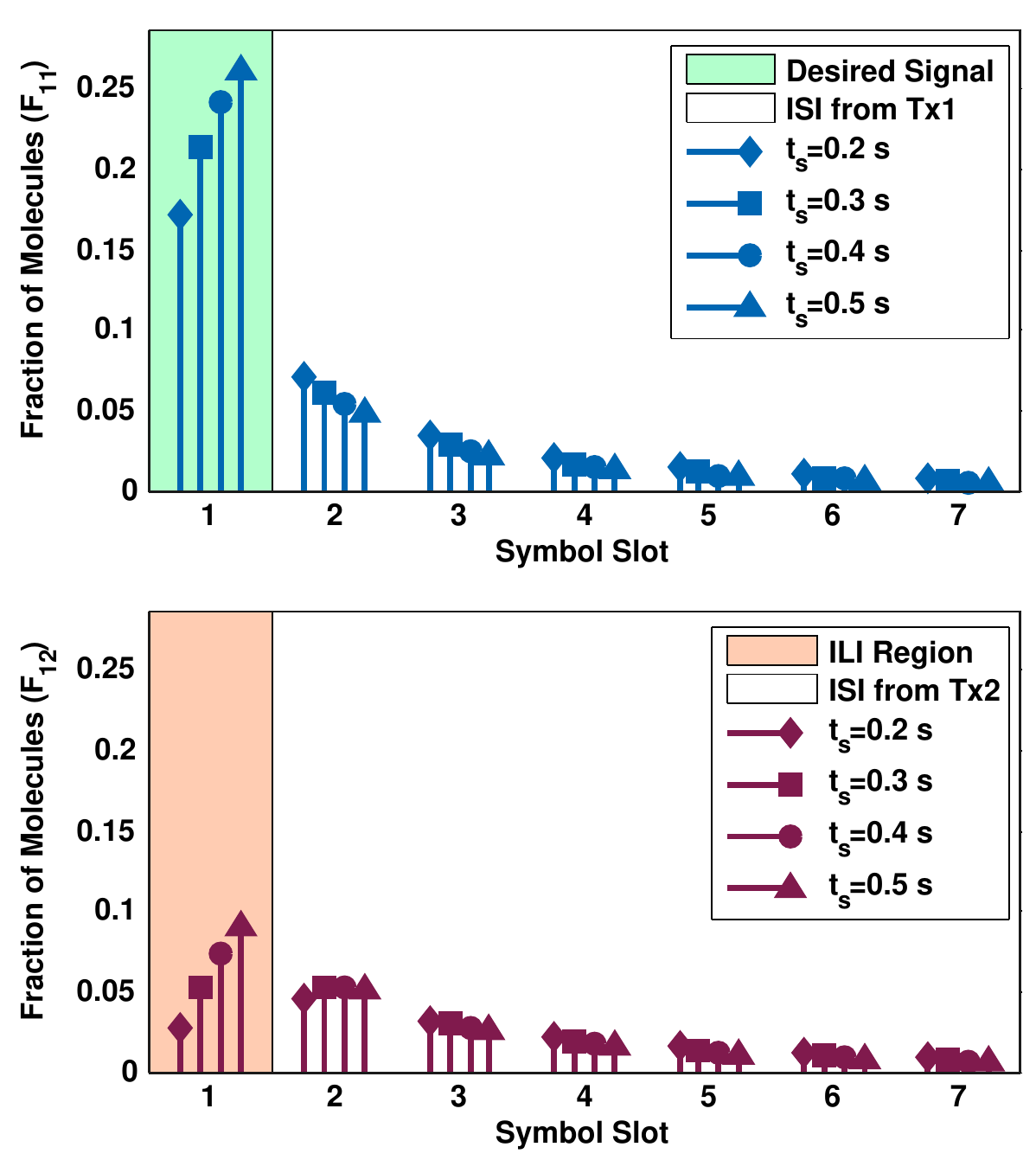}
	\caption{Fraction of hitting molecules for different symbol durations and the cases Tx$_1$/Tx$_2$ to Rx$_1$ via utilizing fitted functions. Top plot contains regions for the desired signal and ISI from Tx$_1$. Bottom plot contains regions for ILI and ISI from Tx$_2$. (\txblue{$D=\SI{50}{\micro\meter^2/\second}$}, \txblue{${h=\SI{1}{\micro\meter}}$}, \txblue{${\rrn=\SI{4}{\micro\meter}}$}).}
	\label{Fig:pks}
\end{figure}
After channel modeling, we can utilize the fitted functions for simulations that consider consecutive transmissions. To do so, we first need to evaluate the fraction of received molecules for consecutive symbol slots. Depending on the symbol duration, we can evaluate the channel response for each symbol slot. Fig.~\ref{Fig:pks} depicts the fraction of received molecules for each consecutive symbol slot. Increasing $t_s$ increases the amplitude of the desired signal and reduces the ISI at the cost of increasing the ILI. Therefore, while determining thresholds, the symbol duration should also be considered. If the detection thresholds are not selected appropriately, then ILI may cause bit-0 to be detected as bit-1 \txblue{or vice versa}, which also means `error'.

\subsection{Detection Algorithms}
\label{Sec:Detection_algorithms}
In this section, we introduce five detection algorithms. Each of them requires a different set of particular pieces of information given to the receiver. Fig.~\ref{Fig:Algorithms} describes the algorithms in terms of the information required. There is a default set that is needed commonly for all the algorithms. The default set consists of system parameters $D$, $t_s$, and topology parameters such as $d$, $h$, and $r_r$. The first algorithm, which we name the \textit{fixed threshold method} works with only the default set. It uses an empirical threshold and does not adapt to varying symbol duration or the signal power. The other four algorithms exploit the statistical analysis of a channel in contrast to the first one. Two of them incorporate \eqref{eq_y_excludingI_inMatrixForm} and the other two utilize~\eqref{eq_y_includingI_inMatrixForm}. 
\begin{figure}
	\centering
	\includegraphics[width=0.99\columnwidth,keepaspectratio]
	{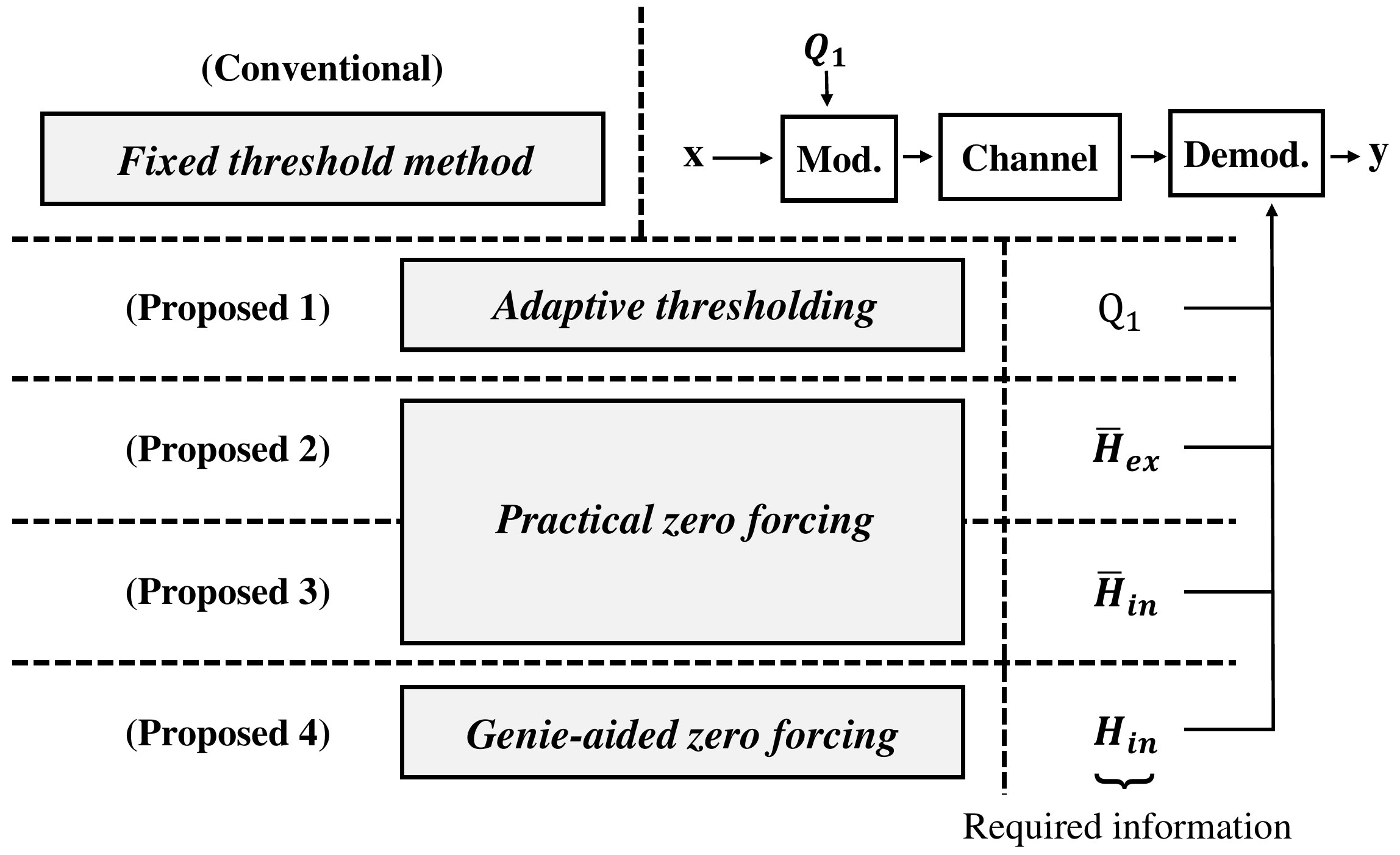}
	\caption{Representation of detection algorithms and their requirements.}
	\label{Fig:Algorithms}
\end{figure}

The second algorithm, which uses the channel model~\eqref{eq_y_excludingI_inMatrixForm}, is called \textit{adaptive thresholding}. It additionally requires $Q_1$, which is not a big assumption since $Q_1$ is determined with the communication protocol and the modulation. The output of the detector, $\hat{\pmb{y}}_{a}$, is formulated in ({\ref{eqn_adaptive}}) and the algorithm calculates the optimal decision threshold accordingly. \txblue{Details of the threshold decision will be presented in Section~\ref{Sec:decision_thresholds}.}
\begin{equation}
\label{eqn_adaptive}
\hat{\pmb{y}}_{a}=\frac{1}{Q_1}\pmb{y}=\frac{1}{Q_1}({\pmb{H}_{\mathrm{ex}}\pmb{x}}+{\pmb{I}_{\mathrm{ex}}}).
\end{equation}
The third algorithm, called \textit{practical zero forcing (ZF) with $\pmb{H}_{\mathrm{ex}}$}, needs the average channel response matrix, which is denoted by $\bar{\pmb{H}}_{\mathrm{ex}}$. The formulation, inspired by the conventional zero forcing communication strategy, is given in ({\ref{eqn_practical}}) and was first proposed in~\cite{koo2015detectionAF}. The output of the detector, $\hat{\pmb{y}}_\text{ex}$, is formulated as 
\begin{align}
\label{eqn_practical}
\hat{\pmb{y}}_\text{ex}&={\bar{\pmb{H}}_{\mathrm{ex}}}^{-1}\pmb{y}={\bar{\pmb{H}}_{\mathrm{ex}}}^{-1}{\pmb{H}_{\mathrm{ex}}\pmb{x}}+{\bar{\pmb{H}}_{\mathrm{ex}}}^{-1}{\pmb{I}_{\mathrm{ex}}}.
\end{align}
It is far from the shape of the conventional zero forcing in RF communication. 

To check the utility of diagonal terms in the channel matrix, we utilize the full rank channel model as in~\eqref{eq_y_includingI_inMatrixForm}. Determining how the exact channel components can be achieved is still an open task; we take the average channel response matrix  denoted by $\bar{\pmb{H}}_{\mathrm{in}}$. The new strategy is named \textit{practical ZF} with { $\pmb{H}_{\mathrm{in}}$} and the formulation is given by
\begin{align}
\label{eqn_more_practical}
\hat{\pmb{y}}_\mathrm{in}&={\bar{\pmb{H}}_{\mathrm{in}}}^{-1}\pmb{y}={\bar{\pmb{H}}_{\mathrm{in}}}^{-1}{\pmb{H}_{\mathrm{in}}\pmb{x}}+{\bar{\pmb{H}}_{\mathrm{in}}}^{-1}{\pmb{I}}.
\end{align}
Lastly, we assume the case that the receiver somehow knows the exact channel state for every signal reception time. Since this is not feasible in practic for this case, we name the algorithm \textit{Genie-aided zero forcing}. The output of the detector, $\hat{\pmb{y}}_\mathrm{g}$, is formulated in ({\ref{eqn_genie}}). Though it provides the best performance, it is not feasible. Indeed, the randomness of the molecular communication channel originates with the Brownian motion of molecules, something that is hard to acquire instantaneously.\footnote{Note that the coherence time of molecular communications is zero unlike RF communications.}
\begin{align}
\label{eqn_genie}
\hat{\pmb{y}}_\mathrm{g}&={\pmb{H}}_{\mathrm{in}}^{-1}\pmb{y}=\pmb{x}+{\pmb{H}_{\mathrm{in}}}^{-1}{\pmb{I}}.
\end{align}

Our first contribution lies in the proof that the \textit{adaptive thresholding} and the \textit{practical ZF} with { $\pmb{H}_{\mathrm{ex}}$} methods perform exactly the same for the symmetrical MIMO topology. 

\begin{theorem}
	\label{theorem_adap&prac}
	When the centers of the transmitter and the receive antennas form a rectangular grid, the detector outputs of the \textit{adaptive thresholding} and the \textit{practical ZF} with { $\pmb{H}_{\mathrm{ex}}$} methods satisfy $\,\hat{\pmb{y}}_\text{a}=A_{0}\,\hat{\pmb{y}}_\mathrm{ex}$, where $A_0$ denotes the probability of a molecule hitting the intended antenna at the current symbol duration.
\end{theorem}
\begin{proof}
	See the Appendix~\ref{sec_app_1}.
\end{proof}
\emph{Theorem}~\ref{theorem_adap&prac} ensures that both methods perform the same on average, since the signal detection properties (i.e., the detector outputs) are similar up to a constant multiplication.

\section{Theoretical Analysis}
\label{Sec:Theory}
\subsection{Interference Formulations}
\label{Sec:threshold}
In this section, we formulate the interference and find, analytically, the optimal  decision rule for the \textit{adaptive thresholding}, \textit{practical ZF} with {$\pmb{H}_{\mathrm{ex}}$}, and \textit{practical ZF} with {$\pmb{H}_{\mathrm{in}}$} methods. To derive the optimal decision threshold, the study uses the maximum-a-posterior (MAP) method.

As the receiver is unaware of parameter $Q_1$ in the case of the \textit{fixed threshold method}, we have to pre-determine a range of $Q_1$ to use and decide the decision threshold, $\eta_{f}$, to minimize the bit error rate (BER) for all $Q_1$. Note that the analysis of \textit{Genie-aided zero forcing} inherits the difficulties of acquiring instantaneous $\pmb{H}_{\mathrm{in}}$. Hence, the optimal thresholds for \textit{Genie-aided zero forcing} must be found empirically.

We should consider the interference when determining the thresholds and consequently the received symbol. The topological symmetry ensures that all the properties of Rx$_1$ and Rx$_2$ coincide in terms of interference, so it is sufficient to analyze Rx$_1$ only. The channel output for Rx$_1$ at the $n^\mathrm{th}$ time slot becomes $y_{\rxi{1}}[n]=\mathcal{S}_{11}[0]x_{1}[n]+{D}_{1}[n]+{I}_{1}[n]$ from ~\eqref{eq_y_Rxi_n}.
We assume ${I}_{1}[n]$ has a mean $\mu_{I}$ and a variance $\sigma_{I}^{2}$. \txblue{We define $A_k$ as the mean value of $\mathcal{S}_{11}[k]$ and $B_k$ as the mean value of $\mathcal{S}_{12}[k]$.} \emph{Lemma}~\ref{lemma_interference} provides the formulations for estimating the mean and the variance of the interference. 

\begin{lemma}
\label{lemma_interference}
The term $\mathcal{S}_{11}[k]{x_{1}[m-k]}$ of~\eqref{eq_y_Rxi_n} has mean value of $\pi_{1}Q_{1}A_{k}$ and variance of ${\pi_{1}Q_{1}A_{k}(1-A_{k})+\pi_{0}\pi_{1}Q_{1}^{2}A_{k}^{2}}$\txblue{, where $\pi_{0}$ and $\pi_{1}$ are a priori probabilities of transmitted bit-0 and bit-1.} 
\end{lemma}
\begin{proof}
With probability $\pi_1$ and $\pi_0$, $\mathcal{S}_{11}[k]x_{1}[m-k]$ follows $\mathcal{N}(Q_{1}A_{k},Q_{1}A_{k}(1-A_{k}))$ and becomes just zero, respectively. Therefore, the mean of the received ISI becomes $\pi_{1}Q_{1}A_{k}$ and the variance becomes
\begin{eqnarray}
\nonumber
\sigma^{2}&=&E\left[\left(\mathcal{S}_{11}[k]x_{1}[m-k]\right)^{2}\right]-{E\left[\mathcal{S}_{11}[k]x_{1}[m-k]\right]}^{2}\\ \nonumber
&=&\pi_{1}(Q_{1}^{2}A_{k}^{2}+Q_{1}A_{k}(1-A_{k}))-\pi_{1}^{2}Q_{1}^{2}A_{k}^{2}\\ \nonumber
&=&\pi_{1}Q_{1}A_{k}(1-A_{k})+(\pi_{1}-\pi_{1}^{2})Q_{1}^{2}A_{k}^{2}\\ \nonumber
&=&\pi_{1}Q_{1}A_{k}(1-A_{k})+\pi_{0}\pi_{1}Q_{1}^{2}A_{k}^{2}.
\end{eqnarray}
\txblue{The second equation comes from the fact that $(Q_{1}^{2}A_{k}^{2}+Q_{1}A_{k}(1-A_{k}))$ is the summation of square of nonzero outputs.}
\end{proof}

Similarly, we can apply \emph{Lemma}~\ref{lemma_interference} to find the mean and variance of the term $\mathcal{S}_{12}[k]{x_{2}[m-k]}$ and sum both to find the total mean and variance of ${I}_{1}[m]$. By \emph{Lemma}~\ref{lemma_interference}, the total interference mean and variance at the $m^{th}$ symbol slot becomes
\begin{align}
\label{eqn_interference_all}
\begin{split}
\mu_{I}&=\pi_{1}Q_{1}\left(\sum\limits_{k=1}^{m-1}A_{k}+B_{k}\right),\\
\sigma_{I}^{2}&=\pi_{0}\pi_{1}Q_{1}^{2}\left(\sum\limits_{k=1}^{m-1}A_{k}^2+B_{k}^2\right)\\
&+\pi_{1}Q_{1}\left(\sum\limits_{k=1}^{m-1}A_{k}(1-A_{k})+B_{k}(1-B_{k})\right)+\sigma_{n}^{2}
\end{split}
\end{align}
where $B_{k}$ denotes the success probability of both $Q_{1}^{-1}\mathcal{S}_{12}[k]$ and $Q_{1}^{-1}\mathcal{S}_{21}[k]$. Applying \emph{Lemma}~\ref{lemma_interference} for the case of $k=0$ leads to the statistics of the desired signal $\mathcal{S}_{11}[0]x_{1}[m]$. Substituting $\mathcal{S}_{11}$ and $Q_{1}A_0$ with $\mathcal{S}_{12}$ and $Q_{1}B_0$ provides the statistics of $D_{1}[m]$--the ILI term. Note that \eqref{eqn_interference_all} does not require the previously sent bit sequences. It  requires only the index of the current symbol, as it evaluates the expected value over the cases. Hence, \eqref{eqn_interference_all} can be used for each symbol consecutively.

After formulating the interference and the detector output, we can now derive the thresholds for the \textit{practical ZF} with { $\pmb{H}_{\mathrm{ex}}$}. We denote the probability density function (PDF) of $\hat{\pmb{y}}_\mathrm{ex}$ when the transmitted bit is $0$ and $1$ as $\hat{\pmb{y}}_{\mathrm{ex}|0}$ and $\hat{\pmb{y}}_{\mathrm{ex}|1}$. Each of them becomes a vector of size $2\times{1}$. The $i^{\mathrm{th}}$ element of the vector represents the number of molecules received at Rx$_i$. From the symmetry of the system, it is enough to observe the statistics of $\hat{\pmb{y}}_{\mathrm{ex}|0}(1)$ and $\hat{\pmb{y}}_{\mathrm{ex}|1}(1)$. By utilizing \eqref{eqn_practical} and \eqref{eqn_b2n}, we can find $\hat{\pmb{y}}_\mathrm{ex}=(Q_{1}A_{0})^{-1}\pmb{y}$. Revisiting the matrix form \eqref{eq_y_excludingI_inMatrixForm} with the result provides the following:
\begin{align}
\label{eq_y_p_1}
\begin{split}
\hat{\pmb{y}}_{\mathrm{ex}}(1)&=(Q_{1}A_{0})^{-1}\cdot\left(\mathcal{S}_{11}[0]x_{1}+\mathcal{S}_{12}[0]x_{2}+I_{1}\right),\\
\hat{\pmb{y}}_{\mathrm{ex}|0}(1)&=(Q_{1}A_{0})^{-1}\mathcal{S}_{12}[0]x_{2}+(Q_{1}A_{0})^{-1}I_{1},\\
\hat{\pmb{y}}_{\mathrm{ex}|1}(1)&=(Q_{1}A_{0})^{-1}\mathcal{S}_{11}[0]+\hat{\pmb{y}}_{\mathrm{ex}|0}(1).
\end{split}
\end{align}

 We assume that $\hat{\pmb{y}}_{\mathrm{ex}|0}(1)$ and $\hat{\pmb{y}}_{\mathrm{ex}|1}(1)$ follow Gaussian distribution $\mathcal{N}\left(\mu_\text{ex0},\sigma_\text{ex0}^2\right)$ and $\mathcal{N}\left(\mu_\text{ex1},\sigma_\text{ex1}^2\right)$ respectively. The parameters in detail can be given by \emph{Lemma} \ref{lemma_interference} and \eqref{eqn_interference_all} as:
\begin{align}
\nonumber
\begin{split}
\mu_\text{ex0}&=\frac{\pi_{1}B_{0}}{A_{0}}\!+\!\frac{{\mu_{I}}}{Q_{1}A_{0}},~\mu_\text{ex1}\!=\!1\!+\!\mu_\text{ex0},\\
\sigma_\text{ex0}^2&=\frac{\pi_{1}B_{0}(1-B_{0})}{A_{0}^{2}Q_{1}}\!+\!\frac{\pi_{0}\pi_{1}B_{0}^{2}}{A_{0}^{2}}\!+\!\frac{{\sigma_{I}^2}}{Q_{1}^{2}A_{0}^2},\\
\sigma_\text{ex1}^2&=\frac{(1-A_0)}{Q_{1}A_0}\!+\!\sigma_\text{ex0}^2.
\end{split}
\end{align}
The terms of \eqref{eqn_interference_all} have different amplitudes and few of them are dominant. Thus, applying CLT for the Gaussian assumption is insufficient. Section~\ref{Sec:Analysis_on_error} offers a precise observation for a suitable distribution, while keeping the assumption to derive a closed form of theoretical thresholds. The feasibility of the thresholds is shown by the simulated results at the end of this paper.

\subsection{Decision Thresholds}
\label{Sec:decision_thresholds}
Now we define a decision rule as $\arg\max(\hat{\pmb{y}}_{\mathrm{ex}|i})$ and need to find the intersection points of the two distributions (i.e., to find the decision threshold, $\eta_{ex}$, for  \textit{practical ZF} with { $\pmb{H}_{\mathrm{ex}}$}). This leads to the equality 
\begin{equation}
\label{eq_decision_equation}
\frac{\sigma_{\mathrm{ex}1}\!\sqrt{2\pi}}{\sigma_{\mathrm{ex}0}\!\sqrt{2\pi}}\mathrm{exp}\left(\!-\frac{(\eta_{\mathrm{ex}}\!-\!\mu_{\mathrm{ex}0})^2}{2\sigma_{\mathrm{ex}0}^2}\!\right)\!=\!\mathrm{exp}\left(\!-\frac{(\eta_{\mathrm{ex}}\!-\!\mu_{\mathrm{ex}1})^2}{2\sigma_{\mathrm{ex}1}^2}\!\right).
\end{equation}
\txblue{Rewriting the form by taking log on both sides, it becomes a quadratic equation and the closed form of the solution provides}
\begin{equation}
\label{eq_eta_p}
\eta_{\mathrm{ex}}=\mu_{\mathrm{ex}0}+\frac{-1\pm\sqrt{1+\left(\beta_{\mathrm{ex}}-1\right)\left(1+\sigma_{\mathrm{ex}0}^{2}\beta_{\mathrm{ex}}\mathrm{ln}\beta_{\mathrm{ex}}\right)}}{\beta_{\mathrm{ex}}-1}
\end{equation}
for $\beta_{\mathrm{ex}}=\left({\sigma_{\mathrm{ex}1}}/{\sigma_{\mathrm{ex}0}}\right)^2>1$. We denote the bigger one as $\eta_{\mathrm{ex}}^+$ and the smaller one as $\eta_{\mathrm{ex}}^-$. Then the decision rule for the decoded bits $\hat{\pmb{x}}$ becomes:
\begin{equation}
\hat{\pmb{x}} = \delta_{\mathrm{ex}}(\hat{\pmb{y}}_{\mathrm{ex}}) = 
\left\{ 
	\begin{array}{ll}
	0 & \eta_{\mathrm{ex}}^{+}>\hat{\pmb{y}_{\mathrm{ex}}}>\eta_{\mathrm{ex}}^{-} \\
	1 & \text{otherwise}
	\end{array}
\right. \nonumber
\end{equation}
where $\delta_{\mathrm{ex}}(\cdot)$ is the decision function at the receiver for \textit{practical ZF} with { $\pmb{H}_{\mathrm{ex}}$}. The solution is supported by observing that $\eta_{\mathrm{ex}}^{+}$ has a negligible gap $\left(\leq10^{-4}\right)$ from the optimal single threshold obtained by Brute-force search with the simulator.

Note that $\beta_{\mathrm{ex}} \geq 1$ because $\sigma_{\mathrm{ex}1}^2$ is the sum of $\sigma_{\mathrm{ex}0}^2$ and the variance of ${\bar{\pmb{H}}_{\mathrm{ex}}}^{-1}\pmb{H}_{\mathrm{ex}}$. The case where $\beta_{\mathrm{ex}}$ becomes $1$ means that $\hat{\pmb{y}}_{\mathrm{ex}|0}$ and $\hat{\pmb{y}}_{\mathrm{ex}|1}$ have the same variances and it is trivial that the threshold becomes $\frac{\mu_{\mathrm{ex}0}+\mu_{\mathrm{ex}1}}{2}$.

We can  similarly obtain the decision rule for \textit{adaptive thresholding}. We define $\hat{\pmb{y}}_{a|0}$ and $\hat{\pmb{y}}_{a|1}$ as in \eqref{eq_y_p_1} and their means and variances, according to \emph{Theorem}~\ref{theorem_adap&prac}, become $A_{0}\mu_{\mathrm{ex}0}$, $(A_{0}\sigma_{\mathrm{ex}0})^2$, $A_{0}\mu_{\mathrm{ex}1}$, and $(A_{0}\sigma_{\mathrm{ex}1})^2$, respectively. Therefore, we have an equation similar to \eqref{eq_decision_equation} for finding the decision threshold $\eta_{a}$ of \textit{adaptive thresholding} method
\begin{equation}
\frac{A_{0}\sigma_{\mathrm{ex}1}}{A_{0}\sigma_{\mathrm{ex}0}}\,\mathrm{exp}\left(-\frac{(\eta_{a}-A_{0}\mu_{\mathrm{ex}0})^2}{2(A_{0}\sigma_{\mathrm{ex}0})^2}\right)=\mathrm{exp}\left(-\frac{(\eta_{a}-A_{0}\mu_{\mathrm{ex}1})^2}{2(A_{0}\sigma_{\mathrm{ex}1})^2}\right) \nonumber
\end{equation}
that can be solved similarly.

In the same way, we define $\hat{\pmb{y}}_{\mathrm{in}|0}$ and $\hat{\pmb{y}}_{\mathrm{in}|1}$ as the received signal of \textit{practical ZF} with { $\pmb{H}_{\mathrm{in}}$} to obtain the decision rule. They have means and variances denoted by $\mu_{in0}$, $\mu_{in1}$, $\sigma_{\mathrm{in}0}^{2}$, and $\sigma_{\mathrm{in}1}^{2}$, respectively. It is enough to find the exact formulation of those parameters, as we can derive the decision rule by repeating the process carried out for the two previous algorithms.

From \eqref{eq_y_includingI_inMatrixForm}, we can achieve the formulation
\begin{equation}
\nonumber
{\bar{\pmb{H}}_{\mathrm{in}}}=Q_{1}\cdot\begin{bmatrix} A_{0}\!&\!B_{0}\\ B_{0}\!&\!A_{0}\end{bmatrix}
\end{equation}
which has full rank if the inequality $A_{0}^{2}>B_{0}^{2}$ holds. The definitions of $A_0$ and $B_0$ provide the formulation
\begin{eqnarray}
\nonumber
\begin{split}
Q_{1}A_{0}&=\bar{\mathcal{S}}_{11}[0]\!=\!\bar{\mathcal{S}}_{11}(0,t_{s})\!=\!F_{11}(0,t_{s}),\\
Q_{1}B_{0}&=\bar{\mathcal{S}}_{21}[0]\!=\!\bar{\mathcal{S}}_{21}(0,t_{s})\!=\!F_{21}(0,t_{s})
\end{split}
\end{eqnarray}
where $F_{11}$ and $F_{21}$ are derived in \eqref{eqn_fpp_model_function} and shown in Fig.~\ref{Fig:ft}. It is a trivial matter that the inequality $A_{0}^{2}>B_{0}^{2}$ holds in any of the cases. Now we are safe to claim that ${\bar{\pmb{H}}_{\mathrm{in}}}$ has its inverse matrix in the closed form of
\begin{equation}
\label{eq_bar_H_in_inverse}
{\bar{\pmb{H}}_{\mathrm{in}}^{-1}}=\frac{1}{Q_{1}\left(A_{0}^{2}-B_{0}^{2}\right)}\cdot\begin{bmatrix} A_{0}\!&\!\!-B_{0}\\\!-B_{0}\!&\!\!A_{0}\end{bmatrix}
\end{equation}
and utilizing \eqref{eq_y_includingI_inMatrixForm} and \eqref{eq_bar_H_in_inverse} results in
\begin{align}
\nonumber
\begin{split}
\hat{\pmb{y}}_{\mathrm{in}}(1)&=\frac{A_{0}\mathcal{S}_{11}[0]-B_{0}\mathcal{S}_{12}[0]}{Q_{1}\left(A_{0}^{2}-B_{0}^{2}\right)}\cdot{x}_{1}\\
&+\frac{A_{0}\mathcal{S}_{21}[0]-B_{0}\mathcal{S}_{22}[0]}{Q_{1}\left(A_{0}^{2}-B_{0}^{2}\right)}\cdot{x}_{2}+\frac{A_{0}I_{1}-B_{0}I_{2}}{Q_{1}\left(A_{0}^{2}-B_{0}^{2}\right)}.
\end{split}
\end{align}
As a result, the following can be achieved:
\begin{align}
\nonumber
\begin{split}
\hat{\pmb{y}}_{\mathrm{in}|0}(1)&=\frac{A_{0}\mathcal{S}_{21}[0]-B_{0}\mathcal{S}_{22}[0]}{Q_{1}\left(A_{0}^{2}-B_{0}^{2}\right)}\cdot{x}_{2}+\frac{A_{0}I_{1}-B_{0}I_{2}}{Q_{1}\left(A_{0}^{2}-B_{0}^{2}\right)},\\
\hat{\pmb{y}}_{\mathrm{in}|1}(1)&=\frac{A_{0}\mathcal{S}_{11}[0]-B_{0}\mathcal{S}_{12}[0]}{Q_{1}\left(A_{0}^{2}-B_{0}^{2}\right)}+\hat{\pmb{y}}_{\mathrm{in}|0}(1).
\end{split}
\end{align}
With \emph{Lemma} \ref{lemma_interference}, the parameters are finally obtained as:
\begin{align}
\nonumber
\begin{split}
\mu_{\mathrm{in}0}&=\frac{A_{0}-B_{0}}{Q_{1}\left(A_{0}^{2}-B_{0}^{2}\right)}\cdot\mu_{I},~\mu_{\mathrm{in}1}\!=\!1\!+\!\mu_{\mathrm{in}0},\\
\sigma_{\mathrm{in}0}^2&=\left(\frac{A_0}{A_{0}^{2}-B_{0}^{2}}\right)^{2}\!\cdot\!\frac{\pi_{1}B_{0}(1-B_{0})}{Q_{1}}\\
&+\left(\frac{B_0}{A_{0}^{2}-B_{0}^{2}}\right)^{2}\!\cdot\!\frac{\pi_{1}A_{0}(1-A_{0})}{Q_{1}}+\frac{A_{0}^{2}+B_{0}^{2}}{Q_{1}^{2}\left(A_{0}^{2}-B_{0}^{2}\right)^2}\cdot\sigma_{I}^{2},\\
\sigma_{\mathrm{in}1}^2&=\frac{A_{0}^{3}(1-A_{0})+B_{0}^{3}(1-B_{0})}{Q_{1}\left(A_{0}^{2}-B_{0}^{2}\right)^{2}}\!+\!\sigma_{\mathrm{in}0}^2.
\end{split}
\end{align}
Substituting $\mu_{\mathrm{ex}0}$, $\sigma_{\mathrm{ex}0}$, and $\sigma_{\mathrm{ex}1}$ with $\mu_{\mathrm{in}0}$, $\sigma_{\mathrm{in}0}$, and $\sigma_{\mathrm{in}1}$ in \eqref{eq_eta_p} leads to the thresholds $\eta_{\mathrm{in}}^{-}$, $\eta_{\mathrm{in}}^{+}$, and the decision rule $\hat{\pmb{x}} = \delta_{\mathrm{in}}\left(\hat{\pmb{y}}_{\mathrm{in}}\right)$ of \textit{practical ZF} with {$\pmb{H}_{\mathrm{in}}$}.

The two distributions (when the transmitted bit is 0 or 1) of the number of received molecules have different variances. Unlike \txblue{a typical system}, this leads to multiple decision thresholds. Sketched in Fig.~\ref{Fig:two_thresholds} is the error region of the case of utilizing multiple thresholds. When the decision is red, it corresponds to $\hat{\pmb{y}}_{X|0}$ for $X$, in this paper, be either `$\mathrm{ex}$' or `$\mathrm{in}$'. The region that is labeled as \emph{Decision: Blue} corresponds to $\hat{\pmb{y}}_{X|1}$. The x-axis values of intersections between the two curves become $\eta_{X}^{-}$ and $\eta_{X}^{+}$ in ascending order, and the areas labeled $E_1$ and $E_2$ correspond to type-I and type-II error probabilities of the algorithm with the decision rule $\hat{\pmb{x}} = \delta_{X}\left(\hat{\pmb{y}}_{X}\right)$.

\begin{figure}
	\centering
	\includegraphics[width=0.99\columnwidth,keepaspectratio]
	{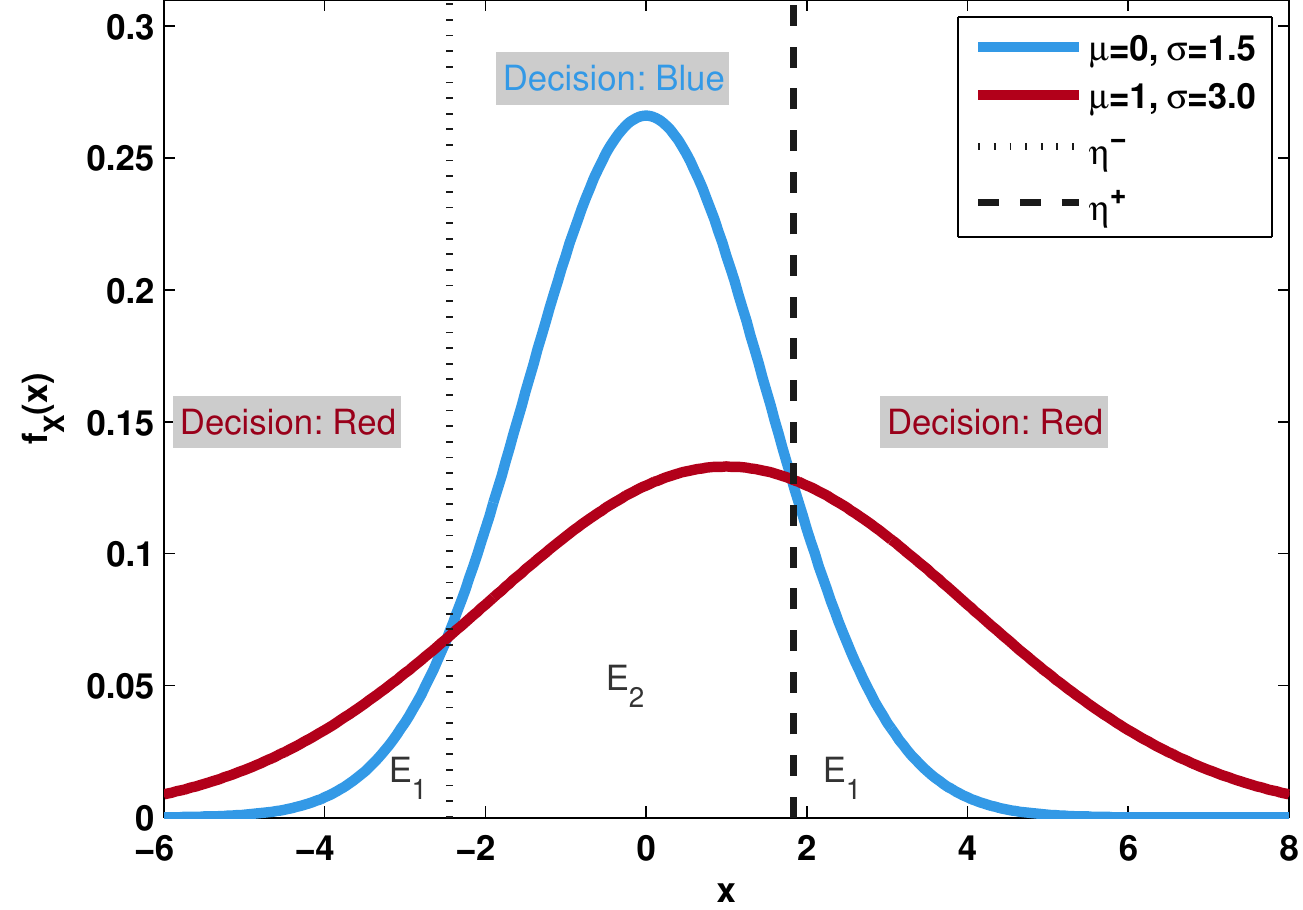}
	\caption{An example case of decision region with multiple thresholds when the mean and the variance of two distributions are different.}
	\label{Fig:two_thresholds}
\end{figure}
\subsection{Error Probability}
\label{Sec:Analysis_on_error}
We define $dt_{Xi}(q)$ as the PDF of detector output at the receive antenna. $X$ is the abbreviation for the detection algorithms such as $\mathrm{ex}$ for \textit{practical ZF} with {$\pmb{H}_{\mathrm{ex}}$} and $\mathrm{in}$ for \textit{practical ZF} with {$\pmb{H}_{\mathrm{in}}$}. The subscript $i$ can be $0$ or $1$, indicating that the conveyed information is bit-0 or bit-1. $dt_{Xi}(q)$ values are assumed to follow the Gaussian distribution $\mathcal{N}\left(\mu_{Xi},\sigma_{Xi}^{2}\right)$ in the previous section, though the simulation output fits poorly with the Gaussian model in the MIMO system, \txblue{as can be observed in Fig.~\ref{Fig:GGD_well_fitted}}. To achieve a more accurate model, we assume that $dt_{Xi}(q)$ follows a generalized Gaussian distribution (GGD) given by
\begin{align}
\begin{split}
dt_{Xi}(q)&=\frac{\beta_{Xi}}{2\alpha_{Xi}\Upgamma(1/\beta_{Xi})}\cdot\mathrm{exp}\left(-\left(|q-\mu_{Xi}|\right)^{\beta_{Xi}}\right),\\
\kappa_{Xi}&=\frac{\Upgamma(5/\beta_{Xi})\Upgamma(1/\beta_{Xi})}{\Upgamma(3/\beta_{Xi})^2},\\
\sigma_{Xi}^{2}&=\frac{\alpha_{Xi}^{2}\Upgamma(3/\beta_{Xi})}{\Upgamma(1/\beta_{Xi})}
\end{split}
\label{eq:scale_parameter}
\end{align}
where $\Upgamma(\cdot)$ denotes the gamma function, $\alpha_{Xi}$, $\beta_{Xi}$ and $\kappa_{Xi}$ denote a scale parameter, shape parameter, and the kurtosis of the distribution. The distribution becomes the Gaussian distribution when the shape parameter equals 2. The simulation results demonstrate that the estimated mean and variance of $dt_{Xi}$ from \txblue{the} Gaussian assumption fit well, but the kurtosis value does not fit. Hence, we need the shape parameter, $\beta$, to generalize the arrival distribution. We do not have an analytical model for $\kappa_{Xi}$. Instead we use the simulation data to find $\beta_{Xi}$. \txblue{The} GGD approximation fits the simulation better than the Gaussian model as shown in Fig.~\ref{Fig:GGD_well_fitted}, for \textit{practical ZF} with {$\pmb{H}_{\mathrm{ex}}$}. The corresponding shape parameters are given in Table~\ref{tab_shape_parameters}.

\begin{figure}[t]
	\centering
	\includegraphics[width=1\columnwidth,keepaspectratio]
	{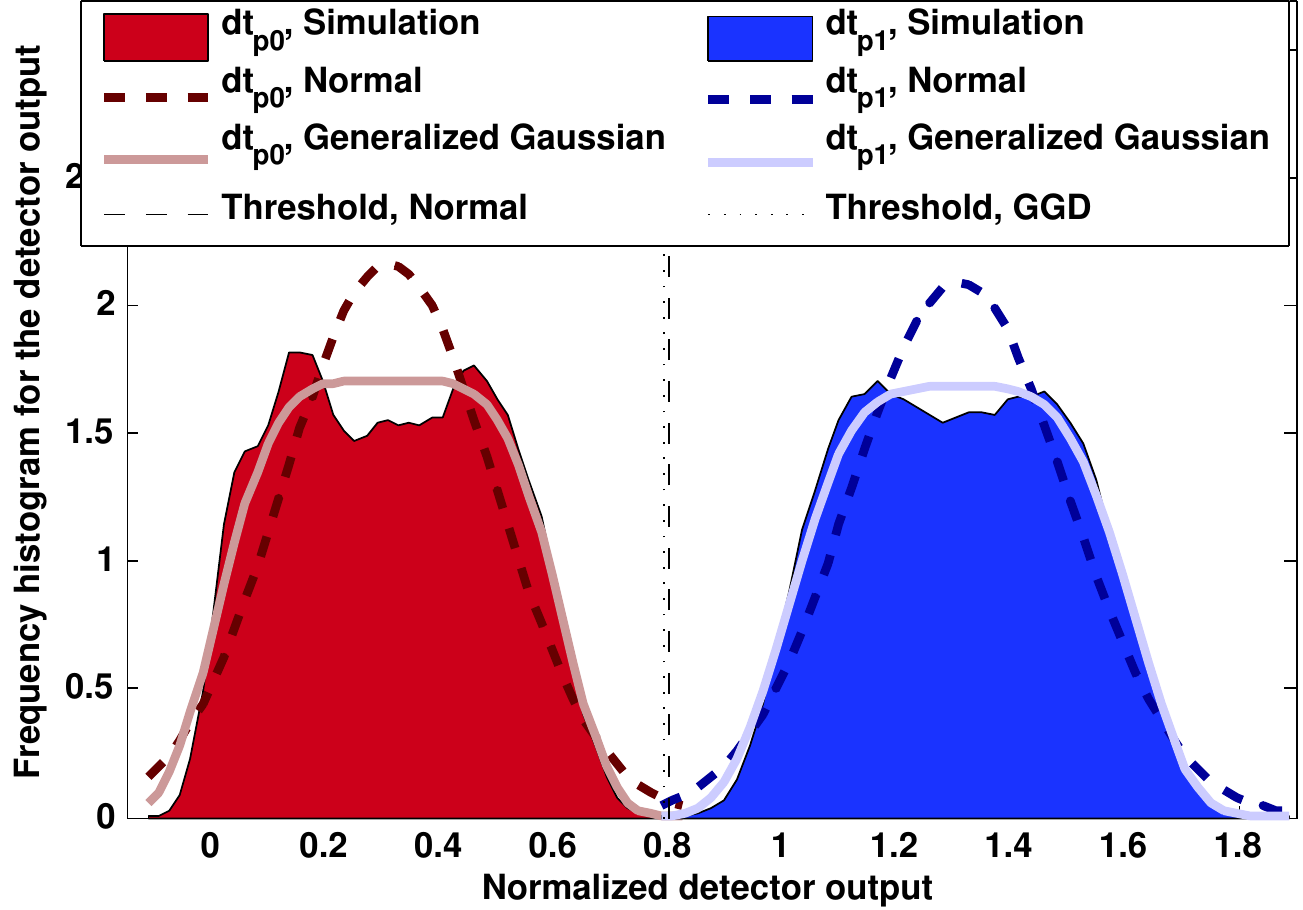}
	\caption{Comparison of simulated and analytical $dt_{pi}$. ($Q_{1}=700$, \txblue{$t_{s}=\SI{0.08}{\second}$}, $\sigma_{n}=10$, $\beta_{\mathrm{ex}0}=5$, $\beta_{\mathrm{ex}1}=4$).}
	\label{Fig:GGD_well_fitted}
\end{figure}
We incorporate the estimated formulations from Section~\ref{Sec:decision_thresholds} with \eqref{eq:scale_parameter} to derive the formulation of error probability. We consider only the upper threshold as the \txblue{below lower threshold error hardly occurs in our system model in practice.} The result becomes:
\begin{align}
\label{eq:Pe}
\begin{split}
& Pe_{X}=f_{\gamma}(\mu_{X0},\alpha_{X0},\beta_{X0})+f_{\gamma}(\mu_{X1},\alpha_{X1},\beta_{X1}),\\
& f_{\gamma}(\mu,\alpha,\beta)=\frac{1}{2}-\frac{\gamma\left(\!\frac{1}{\beta},\left(\frac{|\eta_X^{+}-\mu|}{\alpha}\right)^{\beta}\!\right)}{2\Upgamma\left(\frac{1}{\beta}\right)}
\end{split}
\end{align}
where $\gamma(\cdot)$ denotes the lower incomplete gamma function.



\begin{table}[t!]
	\caption{Shape parameters}
	\label{tab_shape_parameters}
	\renewcommand{\arraystretch}{1.2}
	\centering
	\begin{tabular} {C{0.9cm}|C{1.1cm}C{1.1cm}C{1.1cm}C{1.1cm}}
		\hline
		\multirow{2}{*}{ $\beta$}  & \multicolumn{4}{ c }{$Q_1$}  \\
		& 300 & 500 & 700 & 900 \\ \hline
		$\beta_{p0}$		 & 3.0  & 3.9 & 5.0  & 6.3  \\ 
		$\beta_{p1}$ 		 & 2.57 & 3.3 & 4.0 & 4.5 \\
		\hline
	\end{tabular}
\end{table}

For the following comparison, $t_s$ and the topology parameters are general constant values and $a,b,c$ are coefficients for a second order polynomial.

\begin{theorem}
	\label{theorem_hq1}
	A single positive root exists for function $h(Q_{1})$  that is defined as follows 
	\begin{equation}
	\nonumber
	h(Q_{1})=Q_{1}^{2}\left(\sigma_{\mathrm{ex}0}^{2}-\sigma_{\mathrm{in}0}^{2}\right)=aQ_{1}^{2}+bQ_{1}+c \, 
	\end{equation}
	if we assume that the ILI condition $A_{0}^{2}-2B_{0}^{2}\!>\!0$ and the ISI condition $(A_{0}^{2}-2B_{0}^{2})/3\!>\!\sum_{i=1}^{n}(A_{i}^{2}+B_{i}^{2})$ are satisfied. We define a system satisfying these conditions as having acceptable interference. A priori probabilities of transmitted bit-0 and bit-1 are assumed to be equal (i.e., $\pi_{0}\!=\!\pi_{1}\!=\!0.5$).
\end{theorem}
\begin{proof}
See the Appendix~\ref{sec_app_2}.
\end{proof}

\begin{corollary}
	\label{theorem_prac_with_Hin,Hex}
	The \textit{practical ZF} with { $\pmb{H}_{\mathrm{in}}$} has a lower error probability than that of the \textit{practical ZF} with { $\pmb{H}_{\mathrm{ex}}$} if and only if $Q_1$ is above a certain threshold $\mathcal{T}$.
\end{corollary}
\begin{proof}
	Since $\mu_{X1}\!-\!\mu_{X0}\!=\!1$ holds for both $X\!=\!\mathrm{ex}$ and $X\!=\!\mathrm{in}$, the error probability depends on the variances. It is intuitive that higher variances result in higher probabilities of error. Raising $\sigma_{X0}$ concurrently increases $\sigma_{X1}$. Hence, we can claim that $\sigma_{X0}$ is proportional to $Pe_{X}$. Therefore, to compare the error probability, it is sufficient to compare the variances.
	 
	The variances $\sigma_{\mathrm{ex}0}^2$ and $\sigma_{\mathrm{in}0}^2$ are functions \txblue{of} $Q_1$. Multiplying $Q_1^2$ with both $\sigma_{\mathrm{ex}0}^2$ and $\sigma_{\mathrm{in}0}^2$ yields them to be quadratic polynomials. The difference between them becomes $h(Q_{1})$ in \emph{Theorem~\ref{theorem_hq1}}. Let $\mathcal{T}$ be the positive root of $h(Q_1)$. Therefore $h(Q_1)$ is positive $\forall{Q}_{1}\!>\!\mathcal{T}$ since $a$ is positive (see the Appendix for the proof of \emph{Theorem~\ref{theorem_hq1}}). Positive $h(Q_1)$ implies $\sigma_{\mathrm{ex}0}\!>\!\sigma_{\mathrm{in}0}$ thus $Pe_{\mathrm{ex}}\!>\!Pe_{\mathrm{in}}$. Therefore, the \textit{practical ZF} with { $\pmb{H}_{\mathrm{in}}$} shows a lower error probability than the \textit{practical ZF} with { $\pmb{H}_{\mathrm{in}}$} when $Q_1$ is larger than threshold $z$.
\end{proof}
 
 For the example case of $t_{s}\!=\!\SI{0.08}{\second}$ and $\sigma_{n}\!=\!10$, the computed threshold is $\mathcal{T}\!=\!9.34\times10^{5}$. It is also observed that $\sigma_{n}^{2}$, the noise power, only affects value $c$ in \emph{Theorem~\ref{theorem_hq1}}. Decreasing $\sigma_{n}^{2}$ increases the value $c$ without changing $a$ and $b$, which results in a lower $\mathcal{T}$. This means that if $\sigma_{n}^{2}$ is decreased enough so that $\mathcal{T}$ falls below $Q_1$ \txblue{then} \textit{practical ZF} with { $\pmb{H}_{\mathrm{in}}$} would perform better. As a result, we could claim from \emph{Corollary}~\ref{theorem_prac_with_Hin,Hex} that \textit{practical ZF} with { $\pmb{H}_{\mathrm{in}}$} provides a lower error rate when the signal-to-noise ratio (SNR) is higher than a certain point. We will further investigate the SNR threshold in future work.

\section{Numerical Results}
\label{Sec:results}
The system parameters for theoretical analysis are given in Table~\ref{tab_result_params}.
We first give the definition of signal-to-interference-ratio (SIR) metric in a molecular MIMO system and analyze the effect of  distance, $r_r$, and $h$. Next, we use the BER as the performance metric and analyze the effect of $Q_1$ and $t_s$. 
\begin{table}[h]
	\caption{Range of Parameters Used in the Analysis}
	\label{tab_result_params}
	\centering
	\begin{tabular}{L{3cm} L{1cm} L{3cm}}
		\hline
		Parameter &  Variable &Values \\
		\hline
		Diffusion cefficient  	& $D$  &      		\txblue{$\SI{50}{\micro\meter^2/\second}$} \\
		Distance 				& $d$  &     		\txblue{$\{ 2, 4 \} \si{\micro\meter}$} \\
		Radius of the receiver 	& $r_r$&      		\txblue{$\{ 2, 4 \} \si{\micro\meter}$} \\
		Bulge separation 		& $h$  &      		\txblue{$\{ 1, 2 \} \si{\micro\meter}$} \\
		\# molecules for sending bit-1 	&$Q_1$ &	$\{100\sim\txblue{1000}\}$ molecules \\
		Probability of sending bit-1 & $\pi_1$ & 	0.5 \\
		Symbol duration 		&$t_s$ 		&  		\txblue{$\{ 0.05\sim1 \} \si{\second}$} \\
		Molecular noise variance& $\sigma_{n}^{2}$& 100 \\
		Bit sequence length 	& 	   		&    	$5\times10^{5}$\\
		Replication 			&      		& 		20 \\
		\hline
	\end{tabular}
\end{table}

\subsection{SIR Analysis}
\label{result_SIR}
\txblue{For the definition of SIR, we consider a one-shot signal from Tx$_1$ i.e., a bunch of molecules are emitted at $t\!=\!0$.} SIR is defined as the ratio of the expected number of molecules coming from the intended transmitter in the intended time slot to the mean ILI plus ISI for just a one-shot signal.  
\begin{equation}
\text{SIR}=\frac{F_{11}(0,t_s)}{F_{11}(t_{s},\infty)+F_{12}(0,\infty)}. \nonumber
\end{equation}
Note that this definition is specific to the molecular communication case and explains the clearness of the mean signal term in the received signal. 

\begin{figure}[t]
	\centering
	\includegraphics[width=1\columnwidth,keepaspectratio]
	{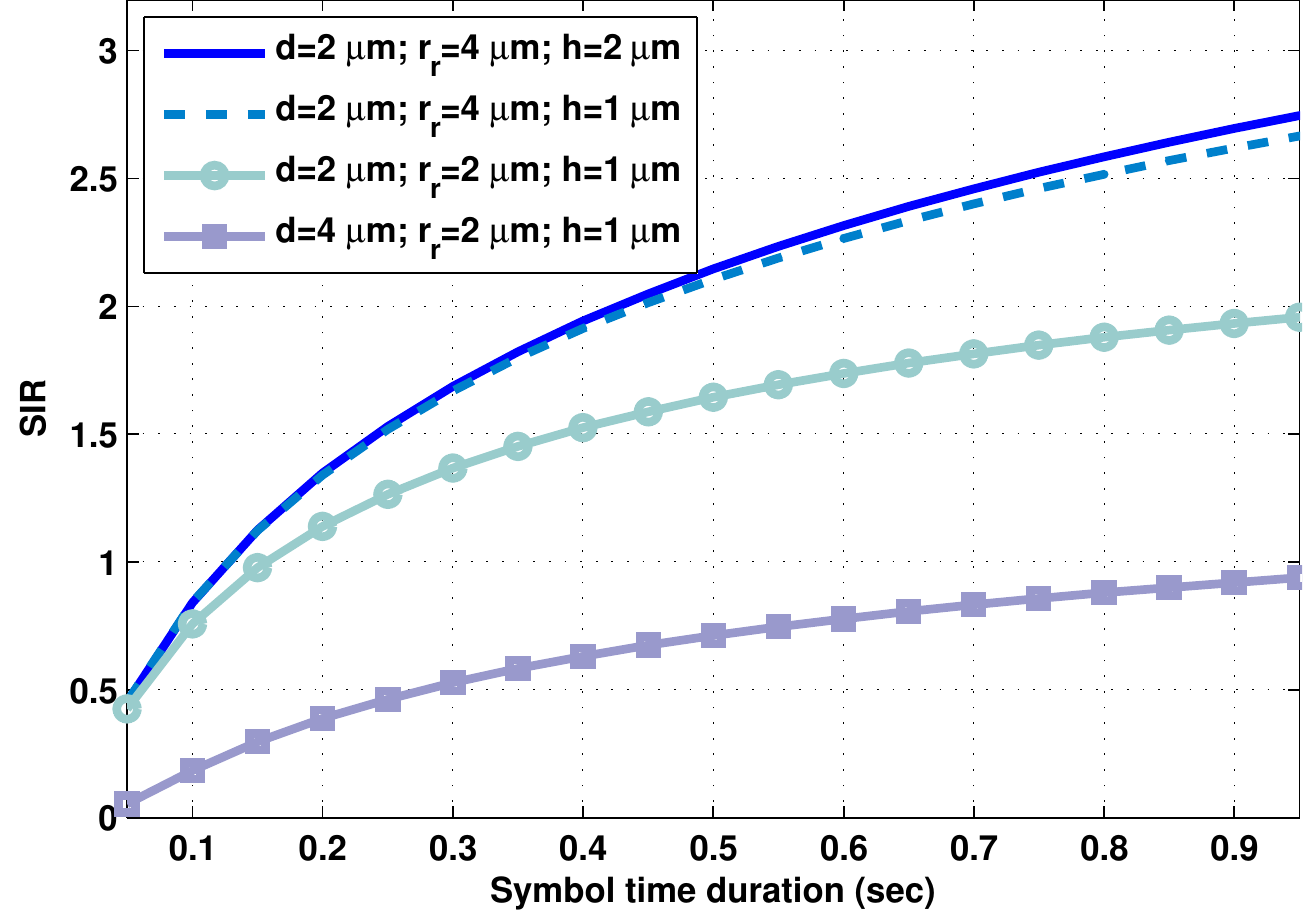}
	\caption{SIR plots for different topologies (\txblue{$D=\SI{50}{\micro\meter^2/\second}$}).}
	\label{Fig:SIR}
\end{figure}
Fig.~\ref{Fig:SIR} shows the SIR values with varying $t_s$ for different topology parameters (such as the distance $d$, radius of the receiver $r_r$, and receiver antenna separation $h$). The best enhancement, in terms of SIR, is observed by reducing the distance between the transmitter and the receiver. Increasing the antenna size also gives merit. For the given parameters, however, a higher $h$ results in a less significant improvement but  non-negligible. Thus, we conclude that for the given parameters the ISI term is more dominant than the ILI term.

For the rest of the performance evaluation, we set the topological parameters as \txblue{$d=\SI{2}{\micro\meter}$}, \txblue{$r_{r}=\SI{4}{\micro\meter}$}, \txblue{$h=\SI{2}{\micro\meter}$}, and \txblue{$D=\SI{50}{\micro\meter^2/\second}$}. The selected system parameters and the fitted values for model parameters are given in Table~\ref{tab_fitting_selected_params_all}. Utilizing the fitted values enables us to estimate $F_{ij}(t)$ analytically.
\begin{table}[h]
	\caption{Fitted model parameters for the selected topology (\txblue{$d=\SI{2}{\micro\meter}$}, \txblue{$r_{r}=\SI{4}{\micro\meter}$}, \txblue{$h=\SI{2}{\micro\meter}$},  \txblue{$D=\SI{50}{\micro\meter^2/\second}$}).}
	\label{tab_fitting_selected_params_all}
	\centering
	\begin{tabular}{l L{1.7cm} L{1.7cm} L{1.7cm}}
		\hline  
		Function                & $b_1$      & $b_2$        & $b_3$     \\
		\hline  
		$F_{11}(t)$          & 0.9155	     & 0.5236       & 0.5476	\\
		$F_{12}(t)$          & 0.2981	     & 0.5315       & 0.5363	\\
		\hline
	\end{tabular}
\end{table}


\begin{figure}[t]
	\centering
	\includegraphics[width=1\columnwidth,keepaspectratio]
	{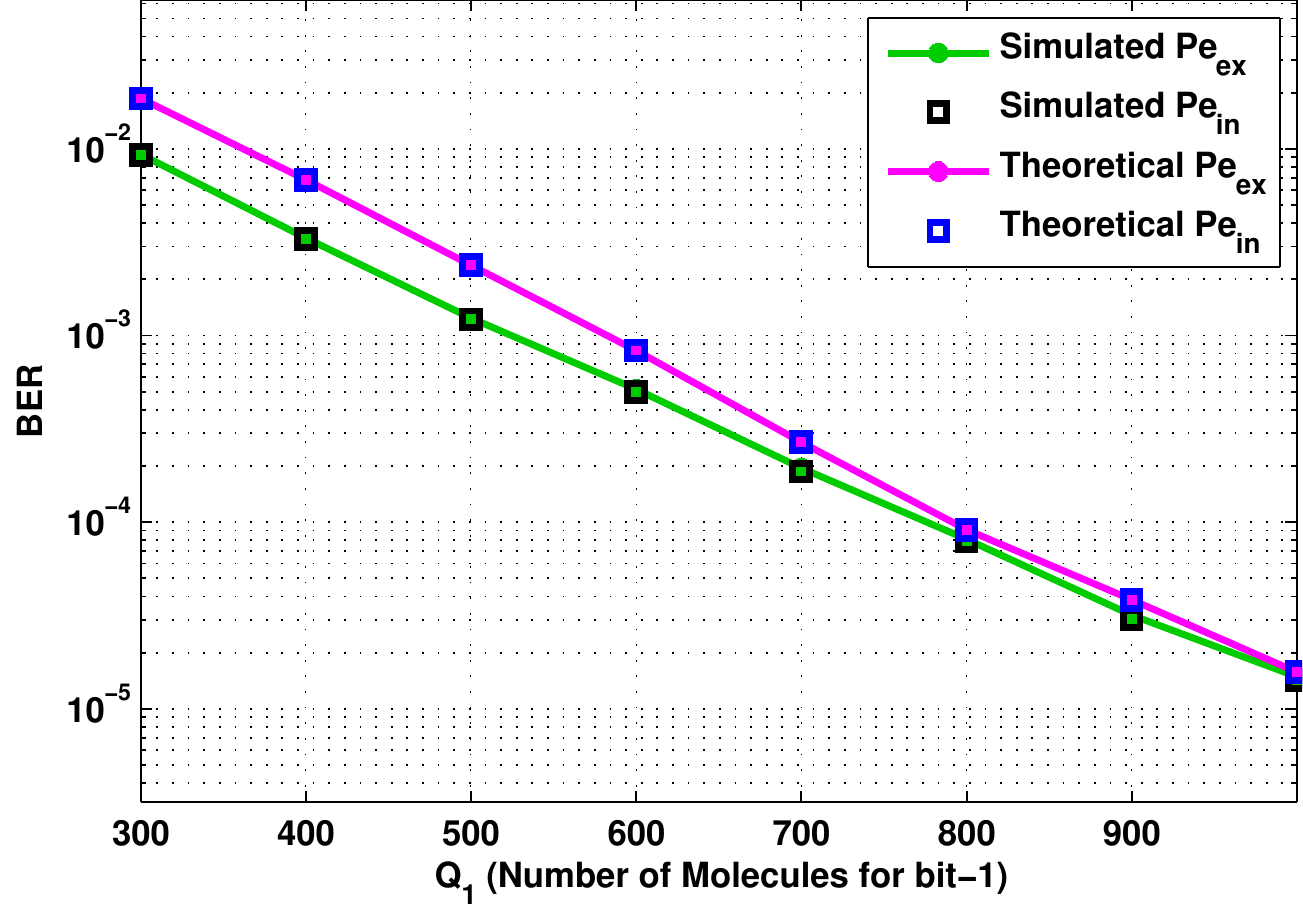}
	\caption{BER performance comparison of the simulation results and the theoretical formulation of the practical zero forcing algorithms (\txblue{$t_s = \SI{80}{\milli\second}$}, $\sigma_{n}=10$).}
	\label{Fig:Sim_vs_Anal}
\end{figure}

\subsection{BER Analysis}
\label{result_BER}
In this section, we analyze the BER with respect to varying $Q_1$ and $t_s$ and compare the performance gain of each of the proposed detection algorithms. Each Tx sends $5\!\times\!10^{5}$ bits with an equal probability of sending bit-1 and bit-0. Most prior work has shown that, with an appropriate symbol duration, the current symbol is affected predominantly by one previous symbol with the rest being \txblue{less dominant}~\cite{kuran2010energyMF, kim2013novelMT, arjmandi2013diffusionBN}. Therefore, we considered in the simulation four slots of interference, \txblue{which is enough to contain all the dominant effects and provides a reasonable run time}. We examined all the thresholds between $0$ and $1$ with a $10^{-3}$ interval and checked the optimal threshold for each $Q_1$. The empirically found fixed threshold $\eta_f$ was selected as $0.2$, \txblue{which is the average of the empirical optimal values for each $Q_1$.}

Figure~\ref{Fig:Sim_vs_Anal} shows the BER performance of \textit{practical ZF} with { $\pmb{H}_{\mathrm{ex}}$} and {$\pmb{H}_{\mathrm{in}}$} in both theory and simulation \txblue{while $Q_1$ varies from 300 to 1000}. As expected, we confirm, through this result, the same performance of two \textit{practical ZF} algorithms. It shows that when the value of $Q_1$ is low (i.e., when the signal power is low) a clear gap emerges between the theory and the simulation. The gap between analytical and numerical results decreases as the power increases, which implies that the GGD approximation is better with a higher transmit power.

\begin{figure}[t]
	\centering
	\includegraphics[width=1\columnwidth,keepaspectratio]
	{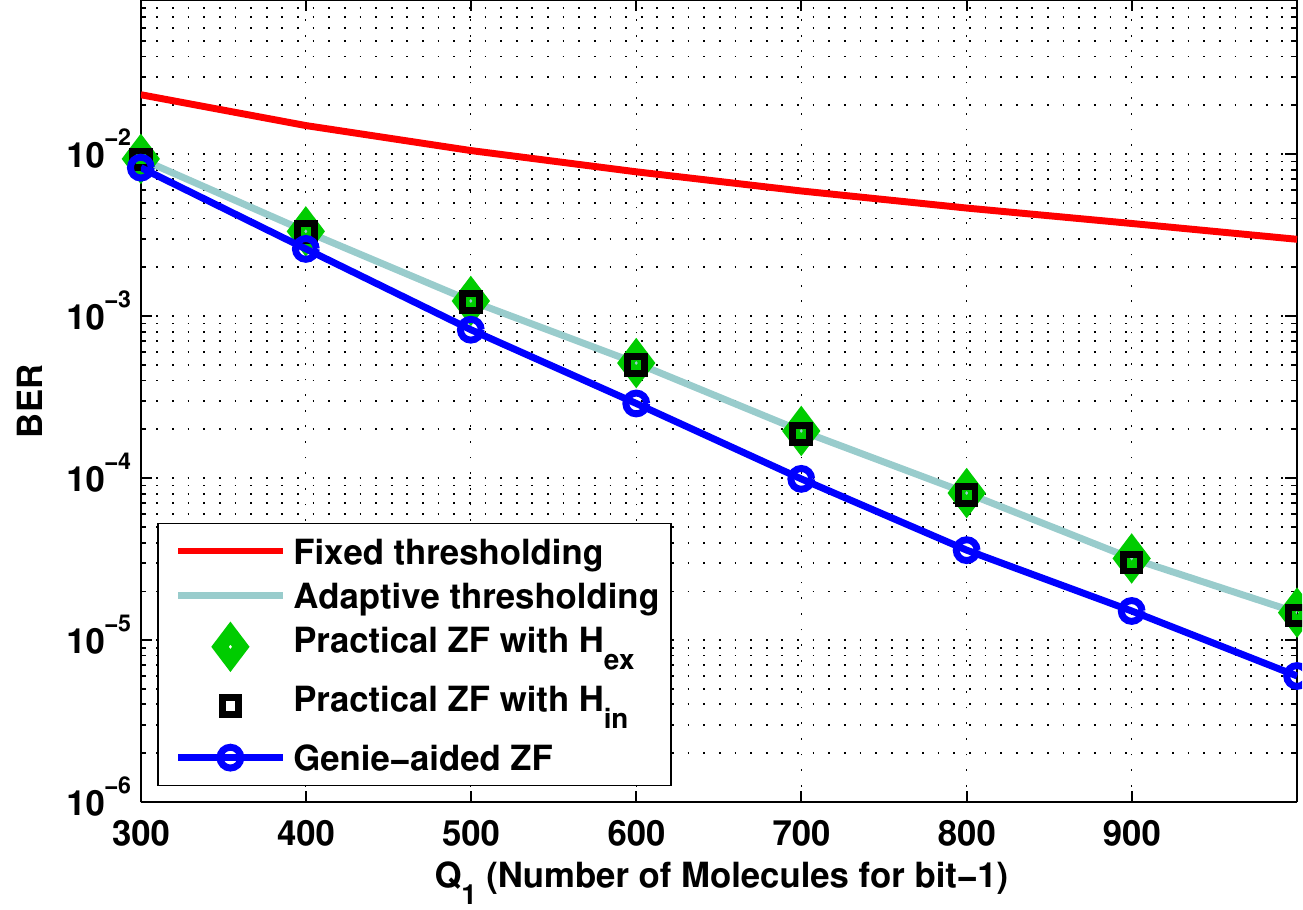}
	\caption{BER performance of the detection algorithms (\txblue{$t_s = \SI{80}{\milli\second}$}).}
	\label{Fig:BER_vs_Q}
\end{figure}
Figure~\ref{Fig:BER_vs_Q} shows the BER performance of the detection algorithms while $Q_1$ varies from 300 to 1000. The first observation is that the \textit{adaptive thresholding} and the \textit{practical ZF} with {$\pmb{H}_{\mathrm{ex}}$} provide the same BER results, as proved in \emph{Theorem~\ref{theorem_adap&prac}}. Increasing $Q_1$ (i.e., the signal power) decreases the BER for all the detection algorithms. However, the improvement of \textit{fixed thresholding} is significantly lower than those of the other methods. When the instantaneous $\pmb{H}$ is known at the receiver, \textit{Genie-aided zero forcing} is applicable and gives the best performance. Note that obtaining that information, however, is not by nature feasible in molecular communications. On the other hand, obtaining the optimal threshold by knowing $Q_1$ and $\pi_1$ is feasible and leads to a performance that is close to \textit{Genie-aided zero forcing}.

\begin{figure}[t]
	\centering
	\includegraphics[width=1\columnwidth,keepaspectratio]
	{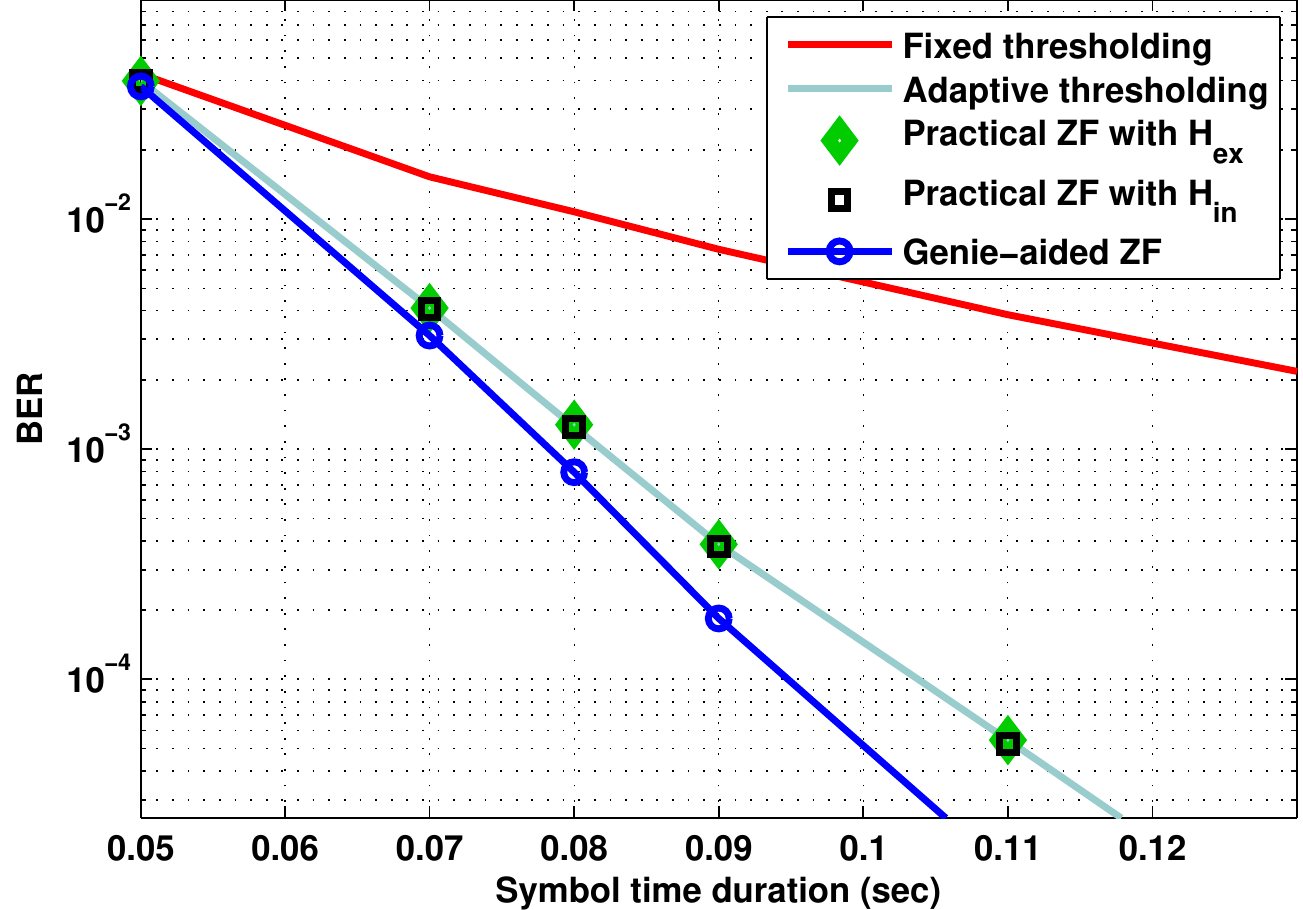}
	\caption{BER performance of the detection algorithms ($Q_{1}=500$).}
	\label{Fig:BER_vs_ts}
\end{figure}
%

Figure~\ref{Fig:BER_vs_ts} illustrates the BER performance against $t_s$ from \txblue{$\SI{50}{\milli\second}$} to \txblue{$\SI{130}{\milli\second}$}. It shows that increasing $t_s$ gives faster improvements in terms of BER relative to increasing $Q_1$. This means that to achieve a lower BER, it is more effective to decrease the information rate than to increase the transmit power. In Fig.~\ref{Fig:Throughput} we compare the throughput of the proposed MIMO and the conventional SISO systems. For this experiment, we used the same amount of emission molecules. As can be seen from the figure, we could confirm that the proposed MIMO system asymptotically achieves double throughput compared to the SISO system. \txblue{Note that the \textit{practical ZF} with {$\pmb{H}_{\mathrm{ex}}$} and the \textit{practical ZF} with {$\pmb{H}_{\mathrm{in}}$} show nearly the same performance in both Fig.~\ref{Fig:BER_vs_Q} and Fig.~\ref{Fig:BER_vs_ts}.}

\section{Testbed}
\label{Sec:Testbed}
We implemented our own macro-scale version to realize the theoretical studies on molecular MIMO systems. The transmitter and the receiver in our system are equipped with multiple transmit nozzles and receive sensors to increase the data rate. The system is low cost, and the testbed platform is modifiable and re-programmable~\cite{infocom15, mobicom15}.

In Fig.~\ref{Fig:Hardware_layout}, the main components of the testbed are shown. The testbed comprises a molecular MIMO transmitter and receiver. The propagation distance between \txblue{the transmitter and the receiver} is about one meter. The transmitter is composed of: 1) a simple user interface for text entry, 2) a microcontroller for executing transmitter algorithms, 3) two reservoirs for chemicals, and 4) two chemical-releasing mechanisms (i.e., two sprays nozzles and an air compressor). At the receiver, the hardware consists of: 1) two chemical sensors for a MIMO system (i.e., two \txblue{MQ-3} alcohol sensors), 2) two microcontrollers that forward an electrical signal that originate from chemical sensors to a computer, and 3) a computer for demodulating, decoding the signal, and visualizing the results.

\subsection{Operation}
Molecular MIMO communication in our testbed is carried out in three steps--encoding, propagation, and detection. In the encoding step, the transmitter separates the input strings into two parts and converts them to bit sequences using international telegraph alphabet no.~2 (ITA2). The ITA2-encoding method is able to convert each letter (among 26 letters) to its corresponding 5-bit binary number (e.g., alphabet `A' = 11000).

\txblue{In the emission step, the transmit antennas (spray nozzles) Tx$_1$ and Tx$_2$ are fed to transmit independent bit sequences that are coded following the encoding step.} The transmitter modulates the bits using binary concentration shift keying (BCSK). The transmit antennas spray once in a symbol duration to transmit bit-1; for bit-0, they do not spray, as the modulation is defined theoretically in Section~\ref{communication}. At the start of communication, two transmit antennas spray once to notify the beginning of transmission. The termination of communication is determined by decoding the bit sequence 00000 at both receive antennas.

\begin{figure}[t]
	\centering
	\includegraphics[width=1\columnwidth,keepaspectratio]
	{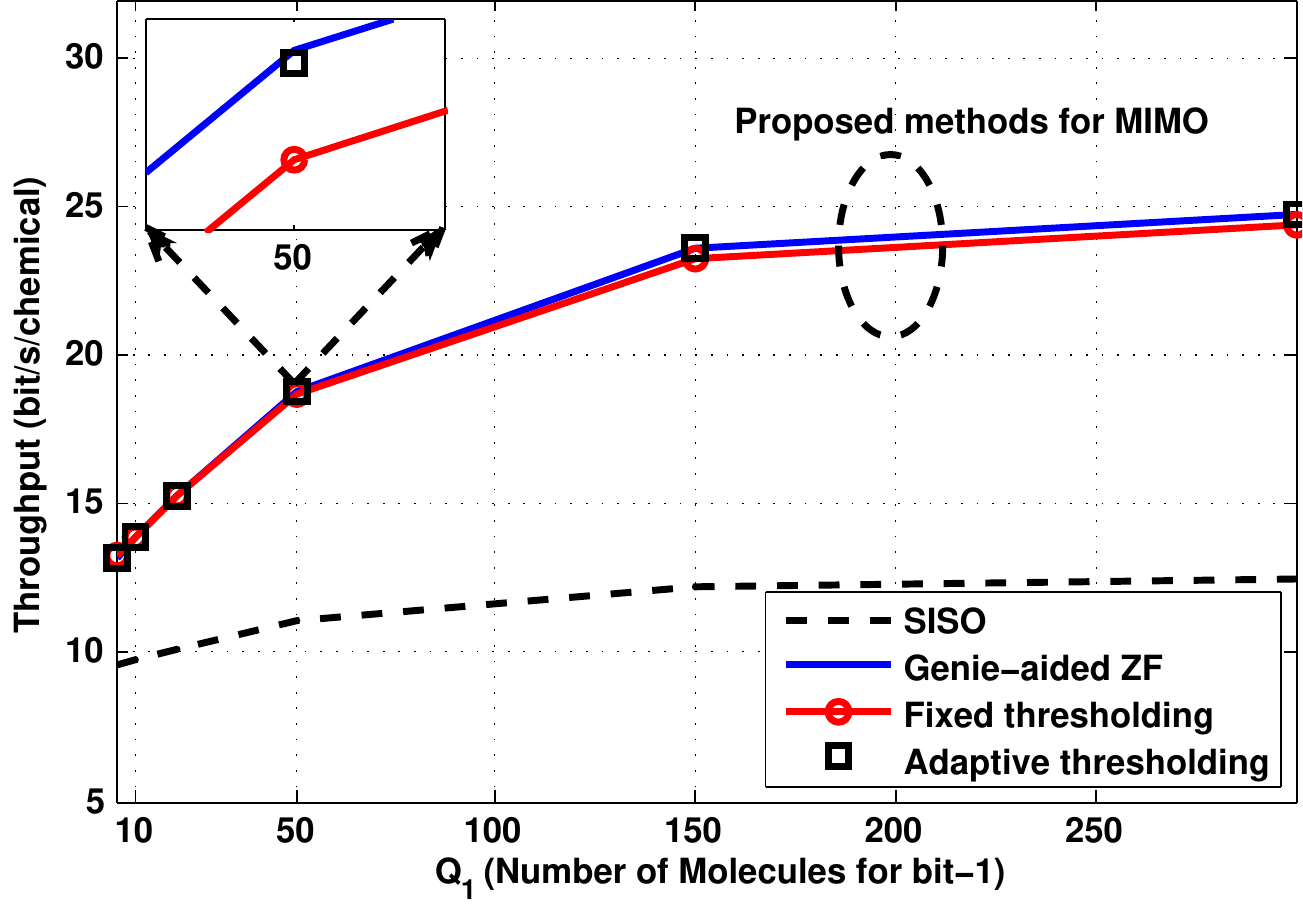}
	\caption{Throughput comparisons (\txblue{$t_s = \SI{80}{\milli\second}$}). $\text{Throughput} = M \cdot  L/{t_s} \cdot (1- \text{BER}) $, where $M$ and $L$ denote the modulation order and the number of spatial streams.}
	\label{Fig:Throughput}
\end{figure}

The emitted molecules diffuse in the air and arrive at the receiver. At the receiver site (i.e., chemical sensors), molecule concentrations are converted to electrical signals (i.e., sensor voltage). \txblue{Voltages are estimated by 250Hz; five of them are averaged to be robust on the error which results in 50Hz of reading frequency.} The coherence time is zero in the diffusion-based molecular communication. Hence we cannot use existing RF communication methods. Therefore, we need algorithms that are designed for diffusion-based molecular communications. Having a lack of control over the number of molecules released, more consistent spray is needed to implement various algorithms. Our testbed has a capability of measuring the reception voltage that is proportional to the number of received molecules. It enables us to use \textit{adaptive thresholding} for the molecular MIMO communication testbed.


\subsection{Hardware Layout}
\begin{figure*}[t]
	\centering
	\includegraphics[width=1.98\columnwidth,keepaspectratio]
	{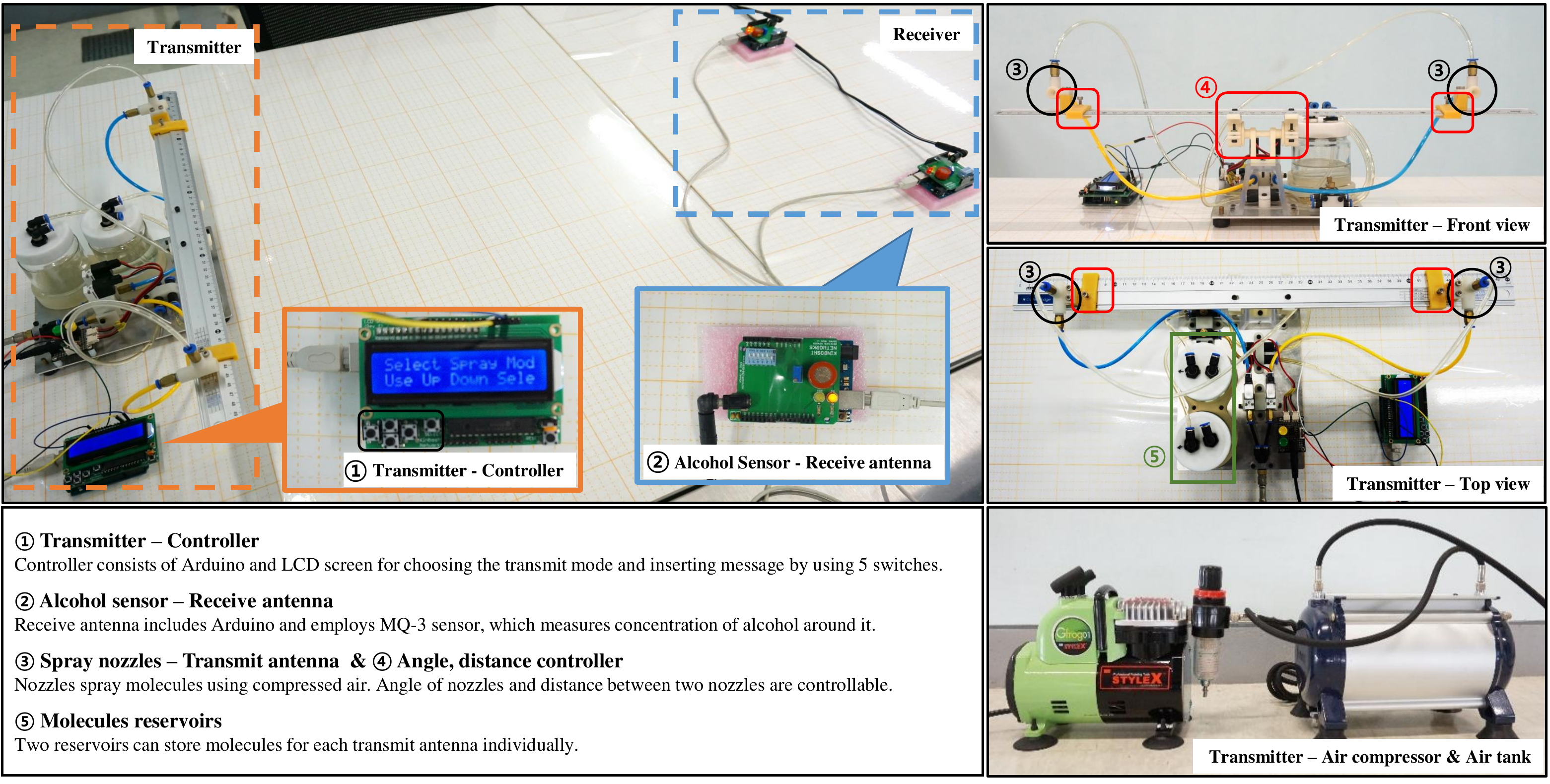}
	\caption{The tabletop molecular MIMO communication platform.}
	\label{Fig:Hardware_layout}
\end{figure*}

\txblue{We have two more implementable algorithms: \textit{fixed thresholding} and \textit{practical ZF} with \textit{$\pmb{H}_{\mathrm{in}}$}. However, to show a clear comparison between them, the transmitter should have the capability of controlling the amounts of emission so that the BER is estimated accordingly. We leave this to the future work with an upgraded version of the testbed and instead focus in this paper on showing the feasibility of the micro-scale strategy on the macro-scale testbed.}

\begin{figure}
	\centering
	\includegraphics[width=1\columnwidth,keepaspectratio]
	{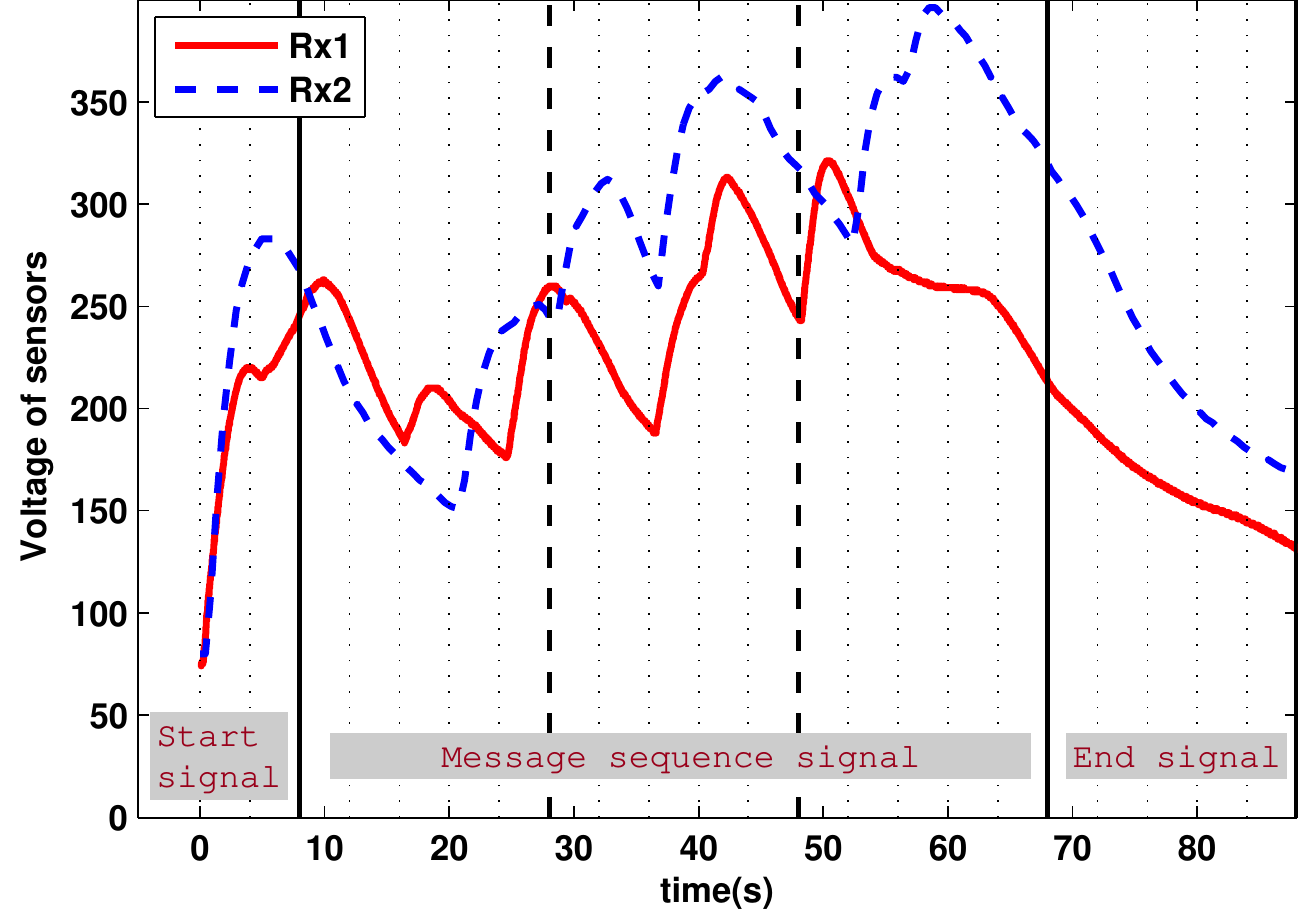}
	\caption{Voltage of sensors when the message `YONSEI' is transmitted.}
	\label{Fig:testbed_reception}
\end{figure}

In Fig.~\ref{Fig:testbed_reception}, the voltage information of each sensor is shown when the message `YONSEI' is sent. Two lines show a change of each receiver sensor's voltage. The first two symbols are the start signal, the next 15 symbols are the message signal, and the final 5 symbols are the end signal. 
\begin{figure}
	\centering
	\includegraphics[width=1\columnwidth,keepaspectratio]
	{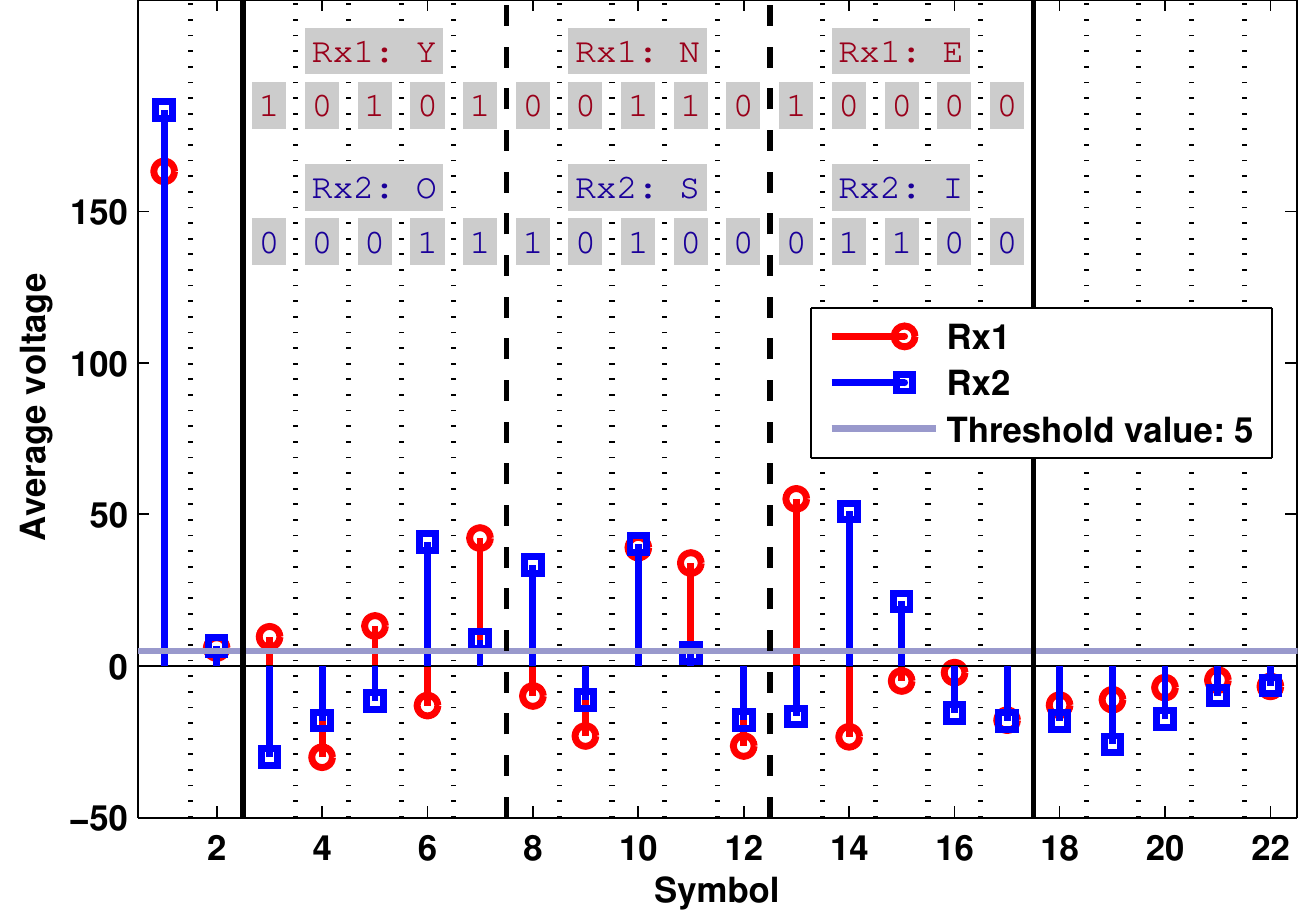}
	\caption{Adaptive thresholding when the message `YONSEI' is sent.}
	\label{Fig:testbed_threshold}
\end{figure}
Fig.~\ref{Fig:testbed_threshold} shows an application of the adaptive thresholding to the signal depicted in Fig.~\ref{Fig:testbed_reception}. After subtracting the first reading in each symbol duration, \txblue{all the sensors' values are averaged respectively over the time of the symbol duration so as to apply the \textit{adaptive thresholding algorithm}}. Then we apply thresholding to the averaged signal value. In Fig.~\ref{Fig:testbed_threshold}, the stem plots with circle and square markers denote data from Rx$_1$ and Rx$_2$. The horizontal thick line indicates the threshold value for the adaptive thresholding algorithm. The threshold value is a bit different from our theoretical analysis shown in the previous sections. This is because the receiver sensors used in our testbed are not ideal and have quite a long turn over time \txblue{($\approx\!100\!\times\!{t}_s$)} to be recovered. Thus for the demonstration, we empirically found the threshold (related to the number of released molecules) for signal detection. 
\begin{table}[t]
	\caption{Data rate and BER comparison}
	\label{tab_data_rate_comparision}
	\centering
	\begin{tabular}{L{0.9cm} L{2.5cm} L{2.1cm} L{1.6cm}}
		\hline
		Type	&    Transmission time (s) & Data rate (bps) & BER\\
		\hline
		SISO	&    148 & 0.20 & $4.412\times10^{-2}$ \\

		MIMO	&    88  &  0.34 & $9.750\times10^{-2}$\\
	
		\hline
	\end{tabular}
\end{table}

\subsection{Demonstration}
For the testbed measurement, we set the distance between transceivers as $90~cm$ and the antenna separation as $40~cm$. \txblue{We used 4 seconds for the symbol duration}. The transmitted message was `YONSEI'. \txblue{Table~\ref{tab_data_rate_comparision} compares the transmission times, the data rates, BER of the SISO and  MIMO systems.} The MIMO system shows a 1.7 times higher data rate than the SISO system. The data rate enhancement is not exactly double due to interference compensation and system overhead (the start and end of communication indicators). \txblue{We also sent consecutive 500 bits each with $\SI{4}{\second}$ symbol duration and we got rough BER values, which is presented in Table~\ref{tab_data_rate_comparision} .}

\txblue{For the MIMO case, the signal-to-ILI (S-ILI) ratio in the current symbol duration is measured by sending bit-1 at Tx$_1$ and bit-0 (silence) at Tx$_2$. We evaluate S-ILI ratio as follows
\begin{align}
\text{S-ILI Ratio} = \sum\limits_{t=0}^{t=t_s}s_{\text{Rx1}}(t) \;/\; \sum\limits_{t=0}^{t=t_s}s_{\text{Rx2}}(t)
\end{align}
where $s_{\text{Rxi}}(t)$ denotes the received signal (i.e., the voltage) at receiver $i$ (Rx$_i$). In other words, we measured the signal at the desired Rx$_i$ and normalize it by ILI term for the current symbol duration. During our tests, we also measured the ${\mbox{S-ILI}}$ and it was $14.567$, which means ILI signal has an average amplitude of nearly the $1/14.567$ of the received signal at the intended antenna in the current symbol slot. }

\section{Conclusions}
\label{Sec:conclusion}
Data rates in molecular communications are affected by interference. Therefore, any enhancement in molecular communications must give precise consideration to interference effects. In this paper, we proposed a MIMO system for MCvD that takes into account inter-symbol and inter-link interference. Moreover, we proposed four symbol detection algorithms that depend on the information set at the receiver. First, we modeled the channel's impulse response by applying least squares curve fitting to the simulation results obtained form our ${\mbox{3-D}}$ MIMO simulator, where the received signal is the fraction of the received molecules. 
We utilized the estimated function (which gives the fraction of received molecules) to 
model interference
and find the optimal thresholds for the 
detection algorithms. In the performance analysis, we investigated the effect on the SIR of varying topological conditions. The results showed that decreasing the transmitter-receiver distance and increasing the size of the receive antennas are more effective at reducing the interference than is increasing the separation of the antennas. We analytically and numerically investigated the link-level performance of the proposed detection algorithms while varying the number of emitted molecules and the symbol duration. We also implemented the world's first molecular MIMO testbed to verify the proposed concept and achieved transmission rates 1.7 times higher than those 
obtained from
the molecular SISO system.


\section*{Acknowledgment}
C.-B. Chae would like to thank S. Kim and C. Kim for their helpful discussions.

\appendix
\subsection{Proof of Theorem~\ref{theorem_adap&prac}}
\label{sec_app_1}

	The topological symmetry guarantees that the random variables $\mathcal{S}_{11}[k]$ and $\mathcal{S}_{22}[k]$ have equal statistical parameters for a positive integer $k$. They are both  binomial random values with a success probability of $A_k$. The transmitter sends $Q_1$ molecules for transmitting a bit-1. Therefore, the diagonal entries of $\pmb{H}_\mathrm{ex}$ follow a binomial distribution and are approximated to the normal distribution as follows:
	\begin{equation}
	\label{eqn_b2n}
	\pmb{H}_{\mathrm{ex}}(i,i) \sim \mathcal{B}(Q_{1},A_0) \approx \mathcal{N}\left(Q_{1}{A}_{0},Q_{1}A_{0}(1-A_{0})\right)
	\end{equation}
	and the channel mean $\bar{\pmb{H}}_{\mathrm{ex}}$ equals $Q_{1}A_{0}\pmb{E}$ where $\pmb{E}$ denotes a $2\times2$ identity matrix. It leads ${\bar{\pmb{H}}}_{\mathrm{ex}}^{-1}=(1/Q_{1}A_{0})\pmb{E}$ and $\hat{\pmb{y}}_{a}$ in ({\ref{eqn_adaptive}}) to become a multiplication of $A_{0}$ and $\hat{\pmb{y}}_{\mathrm{ex}}$ in ({\ref{eqn_practical}}).

\subsection{Proof of Theorem~\ref{theorem_hq1}}
\label{sec_app_2}

The coefficients $a$ and $c$ are expressed in detail as:
\begin{align}
\nonumber
\begin{split}
a&=\frac{B_{0}^{2}}{4A_{0}^{2}}+\frac{\sum_{i=1}^{n}(A_{i}^{2}+B_{i}^{2})}{4}\cdot\left(\frac{1}{A_{0}^{2}}-\frac{A_{0}^{2}+B_{0}^{2}}{(A_{0}^{2}-B_{0}^{2})^2}\right),\\
c&=\sigma_{n}^{2}\left(\frac{1}{A_{0}^{2}}-\frac{A_{0}^{2}+B_{0}^{2}}{(A_{0}^{2}-B_{0}^{2})^2}\right).
\end{split}
\end{align}
$A_{0}\!>\!B_{0}$ is trivial from the definition: hence, the equation
\begin{equation}
\nonumber
\left(\frac{1}{A_{0}^{2}}-\frac{A_{0}^{2}+B_{0}^{2}}{(A_{0}^{2}-B_{0}^{2})^2}\right)=\frac{-B_{0}^{2}(3A_{0}^{2}-B_{0}^{2})}{A_{0}^{2}(A_{0}^{2}-B_{0}^{2})^2}
\end{equation}
provides that $c$ is always a negative value. The \textit{acceptable} interference condition leads to
\begin{align}
\nonumber
\begin{split}
a&>\frac{B_{0}^{2}}{4A_{0}^{2}}-\frac{A_{0}^{2}-2B_{0}^{2}}{3}\cdot\frac{B_{0}^{2}(3A_{0}^{2}-B_{0}^{2})}{4A_{0}^{2}(A_{0}^{2}-B_{0}^{2})^2}\\
&=\frac{B_{0}^{2}}{4A_{0}^{2}}\left(1-\frac{(A_{0}^{2}-2B_{0}^{2})(3A_{0}^{2}-B_{0}^{2})}{3(A_{0}^{2}-B_{0}^{2})^2}\right)\\
&=\frac{B_{0}^{2}}{4A_{0}^{2}}\cdot\frac{A_{0}^{2}B_{0}^{2}+B_{0}^{4}}{3(A_{0}^{2}-B_{0}^{2})^2}>0.
\end{split}
\end{align}
In terms of $Q_1$, $h(Q_{1})=0$ has only one positive root since $a$ is positive and $c$ is negative.

\bibliographystyle{IEEEtran}
\bibliography{references_molcom}

\end{document}